\documentclass[journal]{IEEEtran}
\usepackage{amssymb,amsfonts,amsmath,amsthm,amscd,bbold,dsfont,mathrsfs,color}
\usepackage{graphicx}
\usepackage[tight]{subfigure}
\usepackage{cite}
\DeclareMathAlphabet{\mathpzc}{OT1}{pzc}{m}{it}

\newtheorem{propo}{Proposition}[section]
\newtheorem{lemma}[propo]{Lemma}
\newtheorem{definition}[propo]{Definition}

\newtheorem{thm}[propo]{Theorem}

\newcommand{\sign}{\text{sign}}

\def\E{\mathds E}
\def\1{\mathds 1}

\def\npi{{\rm npi}}

\ifCLASSINFOpdf
\else
\fi
\begin{document}
%
\title{Asymptotic Analysis of Complex LASSO via Complex Approximate Message Passing (CAMP)}
%
%
%

\author{Arian~Maleki,
        Laura~Anitori,
        Zai Yang,
        and~ Richard~Baraniuk,~\IEEEmembership{Fellow,~IEEE}
\thanks{Arian Maleki is with the Department of Statistics, Columbia University, New York city, NY (e-mail:arian.maleki@rice.edu).}
\thanks{Laura Anitori is with TNO, The Hague, The Netherlands (e-mail: laura.anitori@tno.nl)}
\thanks{Zai Yang is EXQUISITUS, Center for E-City, School of Electrical and Electronic Engineering, Nanyang Technological University, Singapore (e-mail: yang0248@e.ntu.edu.sg)}
\thanks{Richard Baraniuk is with the Department of Electrical and Computer Engineering, Rice University, Houston, TX (e-mail: richb@rice.edu).}
\thanks{Manuscript received August 1, 2011; revised July 1, 2012.}}

%
%

\markboth{IEEE Transactions on Information Theory,~Vol.~.., No.~.., ..}%
{}
%



\maketitle

\begin{abstract}
\noindent{Recovering a sparse signal from an undersampled set of
random linear measurements is the main problem of interest in
compressed sensing. In this paper, we consider the case where both
the signal and the measurements are complex-valued. We study the popular
recovery method of $\ell_1$-regularized least squares or
LASSO. While several studies have shown that LASSO
provides desirable solutions under certain conditions, the precise
asymptotic performance of this algorithm in the complex setting is
not yet known. In this paper, we extend the approximate
message passing (AMP) algorithm to solve the complex-valued LASSO problem and 
obtain the complex approximate message passing
algorithm (CAMP). We then generalize the state evolution framework
recently introduced for the analysis of AMP to the complex setting.
Using the state evolution, we derive accurate formulas for the phase
transition and noise sensitivity of both LASSO and CAMP. Our theoretical results are concerned with the case of i.i.d. Gaussian sensing
matrices. Simulations confirm that our results hold for a larger class of random matrices.  }
\end{abstract}

\begin{IEEEkeywords}
compressed sensing, complex-valued LASSO, approximate message passing, minimax analysis.
\end{IEEEkeywords}

%
\IEEEpeerreviewmaketitle

\section{Introduction}\label{sec:intro}

Recovering a sparse signal from an undersampled set of random linear
measurements is the main problem of interest in compressed sensing (CS). In
the past few years many algorithms have been proposed for signal
recovery, and their performance has been analyzed both analytically
and empirically \cite{TrWr10,ChDoSa98,DoMaMo09,BeFr08,YaZhDeLu11,MaDo09sp}. However, whereas most of the
theoretical work has focussed on the case of real-valued signals and
measurements, in many applications, such as magnetic resonance imaging and radar,
the signals are more easily representable in the complex domain \cite{Lustig, AnMaBa12, BaSt07, HeSt09}.
In such applications, the real and imaginary
components of a complex signal are often either zero or non-zero simultaneously.
Therefore, recovery algorithms may benefit from this prior
knowledge. Indeed the results presented in this paper confirm this
intuition.

Motivated by this observation, we investigate the performance
of the complex-valued LASSO in the case of noise-free and noisy
measurements. The derivations are based on the state
evolution (SE) framework, presented previously in \cite{DoMaMo09}.
Also a new algorithm, complex approximate message passing
(CAMP), is presented to solve the complex LASSO problem. This
algorithm is an extension of the AMP algorithm
\cite{DoMaMo09,DoMaMoNSPT}. However, the extension of AMP and its
analysis from the real to the complex setting
is not trivial; although CAMP shares
 some interesting features with AMP, it is substantially more challenging to establish the characteristics of CAMP. Furthermore, some important features of CAMP are specific to complex-valued signals and the relevant optimization problem. Note that the extension of the Bayesian-AMP algorithm to complex-valued signals has been considered elsewhere \cite{ViSc12, SoPoSc10} and is not the main focus of this work.      

In the next section, we briefly review some
of the existing algorithms for sparse signal recovery in the
real-valued setting and then focus on recovery algorithms for the
complex case, with particular attention to the AMP and CAMP
algorithms.  We then introduce two criteria which we use as
measures of performance for various algorithms in noiseless and noisy settings.
Based on these criteria, we establish the novelty of our results compared to the existing
work. An overview of the organization of the rest of the paper is provided in
Section \ref{sec:organization}.

\subsection{Real-valued sparse recovery algorithms}
Consider the problem of recovering a sparse vector $s_o \in
\mathds{R}^N$ from a noisy undersampled set of linear measurements
$y \in \mathds{R}^n$, where $y = A s_o+ w$ and $w$ is the noise. Let
$k$ denote the number of nonzero elements of $s_o$. The measurement
matrix  $A$ has i.i.d. elements from a given distribution on
$\mathds{R}$. Given $y$ and $A$, we seek an approximation to $s_o$. 

Many recovery
algorithms have been proposed, ranging from convex relaxation techniques to
greedy approaches to iterative thresholding schemes. See \cite{TrWr10} and the
references therein for an exhaustive list of algorithms. \cite{MaDo09sp} has compared several different
recovery algorithms and concluded that among the algorithms compared in that paper
the $\ell_1$-regularized least squares, a.k.a. LASSO or BPDN \cite{TibLASSO, ChDoSa98} that seeks the minimizer of
$\min_x \frac{1}{2} \|y-Ax\|_2^2+ \lambda \|x\|_1$ provides the best performance in the
sense of the sparsity/measurement tradeoff. Recently, several iterative thresholding algorithms have
been proposed for solving LASSO using few computations per-iteration; this enables the use of the
LASSO in high-dimensional problems. See \cite{MaBa11} and the references therein for an exhaustive list of these algorithms.

In this paper, we are particularly interested in
AMP \cite{DoMaMo09}. Starting from $x^0 =0$ and $z^0= y$, AMP uses the following iterations:
\begin{eqnarray*}
x^{t+1} &=& \eta_{\circ}\big(x^t+ A^T z^t;\tau_t\big),\\
z^t &=& y- Ax^t + \frac{|I^t|}{n}z^{t-1},
\end{eqnarray*}
where $\eta_{\circ}(x; \tau) =  (|x|- \tau)_+\sign(x)$ is the soft thresholding function, $\tau_t$ is the threshold parameter, and $I^t$ is the active set of $x^t$, i.e., $I^t =\{i \ | \ x^t_i \neq 0 \}$. The notation $|I^t|$ denotes the cardinality of $I^t$. As we will describe later, the strong connection between AMP and LASSO and the ease of predicting the performance of AMP has led to an accurate performance analysis of LASSO \cite{DoMaMoNSPT}, \cite{BaMo10}.

\subsection{Complex-valued sparse recovery algorithms}
Consider the complex setting, where the signal $s_o$, the measurements $y$, and the matrix $A$
are complex-valued. The success of LASSO
has motivated researchers to use similar techniques in this setting as well. We consider the following two schemes that have
been used in the signal processing literature:
\begin{itemize}
\item r-LASSO: The simplest extension of the LASSO to the complex setting is to consider the complex
signal and measurements as a $2N$ dimensional
real-valued signal and $2n$ dimensional real-valued measurements, respectively. Let
the superscript ${R}$ and $I$ denote the real and imaginary parts of
a complex number. Define $\tilde{y} \triangleq [(y^R)^T, (y^I)^T]^T$ and
$\tilde{s_o} \triangleq [(s_o^R)^T, (s_o^I)^T]^T$, where the superscript $T$
denotes the transpose operator. We have
\[ \tilde{y} =\underbrace{\left( \begin{array}{cccc}
A^R &  -A^I\\
A^I & A^R
\end{array} \right)}_{\tilde{A} \triangleq} \tilde{s}_o.\]
We then search for an approximation of $\tilde{s}_o$ by
solving $\arg \min_{\tilde{x}} \frac{1}{2}\|\tilde{y}-\tilde{A}\tilde{x}\|_2^2+
\lambda  \|\tilde{x}\|_1$ \cite{TaHl08,PiWe09}. We call this algorithm
r-LASSO. The limit of the solution as $\lambda \rightarrow 0$ is
\[
\arg \min_{\tilde x} \|\tilde{x}\|_1, \ \  {\rm s.t.} \  \tilde{y} =
\tilde{A}{\tilde{x}},
\]
which is called the basis pursuit problem, or r-BP in this paper. It is
straightforward to extend the analyses of
LASSO and BP for the real-valued signals to r-LASSO and r-BP.\footnote{The
asymptotic theoretical results on LASSO and BP consider i.i.d. Gaussian
measurement matrices \cite{DoTa09}. However, it has been conjectured that the results
are universal and hold for a ``larger" class of random matrices
\cite{DoTa09b,DoMaMoNSPT}.} 

r-LASSO ignores the
information about any potential grouping of the real and imaginary parts. But, in many applications the real and imaginary components tend to
be either zero or non-zero simultaneously. Considering this extra information in the recovery stage may improve the overall performance of a CS system.

\item c-LASSO: Another natural extension of the LASSO to the complex setting is the following optimization problem that we term c-LASSO
\begin{eqnarray*}
&& \min \frac{1}{2}\|y-Ax\|_2^2+\lambda \|x\|_1,
\end{eqnarray*}
where the complex $\ell_1$-norm is defined as $\|x\|_1 \triangleq \sum_i |x_i| =\sum_i
\sqrt{(x_i^R)^2 + (x_i^I)^2}$
\cite{BeFr08,FiNoWr07,WrNoFi08,Boyd,YaZhDeLu11}. The limit of the solution as $\lambda \rightarrow 0$ is
\[
\arg \min_{ x} \|x\|_1, \ \  {\rm s.t.} \  {y} = {A}{{x}},
\]
which we refer to as c-BP.
\end{itemize}
An important question we address in this paper is: can
we measure how much the grouping of the real and the imaginary parts
improves the performance of c-LASSO compared to r-LASSO? Several papers have considered
similar problems  \cite{huang2009benefit,  peng2010regularized,
duarte2011performance, van2010theoretical, chen2006theoretical,
eldar2010block, lvgroup, stojnic2009reconstruction, stojnic2010ell,
stojnic2009block, ji2009multi, BakinThesis, MeVaBu08, YuLi06,
BachConsistency, NaRi08, LoPoTsGe09, MaCeWi05} and have provided
guarantees on the performance of c-LASSO. However, the
results are usually inconclusive because of the loose constants
involved in the analyses.
This paper addresses the above questions with an analysis that does not involve any loose constants and therefore provides accurate comparisons.\\

 Motivated by the recent results in the asymptotic analysis of the
LASSO \cite{DoMaMo09}, \cite{DoMaMoNSPT}, we first derive the
complex approximate message passing algorithm (CAMP) as a fast and
efficient algorithm for solving the c-LASSO problem. We then extend
the state evolution (SE) framework introduced in \cite{DoMaMo09} to
predict the performance of the CAMP algorithm in the asymptotic
setting. Since the CAMP algorithm solves c-LASSO, such predictions
are accurate for c-LASSO as well for $N \rightarrow \infty$. The analysis carried out in this paper provides new
information and insight on the performance of the c-LASSO that was
not known before such as the least favorable distribution and the
noise sensitivity of c-LASSO and CAMP. A more detailed description
of the contributions of this paper is summarized in Section
\ref{sec:contribution}.

 \subsection{Notation}\label{sec:notation}
 Let $|\alpha|$, $\measuredangle \alpha$, $\alpha^*$, $\mathcal{R}(\alpha)$, $\mathcal{I}(\alpha)$ denote the amplitude, phase, conjugate, real part, and imaginary part of $\alpha \in \mathds{C}$ respectively. Furthermore, for the matrix $A \in \mathds{C}^{n \times N}$, $A^*$, $A_\ell $, $A_{\ell j}$ denote the conjugate transpose, $\ell^{\rm th}$ column and $\ell j^{\rm th}$ element of matrix $A$. We are interested in approximating a sparse signal $s_o \in \mathds{C}^N$ from an undersampled set of noisy linear measurements $y = As_o+w$. $A \in \mathds{C}^{n \times N}$ has i.i.d. random elements (with independent real and imaginary parts) from a given distribution that satisfies $\E A_{\ell j} =0 $ and $\E|A_{\ell j}|^2 = \frac{1}{n}$, and $w \in \mathds{C}^N$ is the measurement noise. Throughout the paper, we assume that the noise is i.i.d. $CN(0, \sigma^2)$, where $CN$ stands for the complex normal distribution. \\

 We are interested in the asymptotic setting where $\delta = n/N$ and $\rho = k/n$ are fixed, while $N \rightarrow \infty$. We further assume that the elements of $s_o$ are i.i.d. $s_{o,i} \sim (1- \rho \delta) \delta_0(|s_{o,i}|) + \rho \delta G(s_{o,i})$, where $G$ is an unknown probability distribution with no point mass at $0$, and $\delta_0$ is a Dirac delta function.\footnote{This assumption is not necessary and as long as the marginal distribution of $s_o$ converges to a given distribution the statements of this paper hold. For further information on this, see \cite{DoMaMoNSPT} and \cite{BaMo10}.} Clearly, the expected number of non-zero elements in the vector $s_o$ is $\rho \delta N$. We call this value the {\em sparsity level} of the signal. In this model, we are assuming that all the non-zero real and imaginary coefficients are paired. This quantifies
 the maximum amount of improvement the c-LASSO gains by grouping the real and imaginary parts. \\
 
 We use the notations $\E$, $\E_X$, and $\E_{X \sim F}$ for expected value, conditional expected value given the random variable $X$, and expected value with respect to a random variable $X$ drawn from the distribution $F$, respectively. Define $\mathcal{F}_{\epsilon, \gamma}$ as the family of distributions $F$ with $\E_{X \sim F} (\mathds{I} (|X|= 0)) \geq 1- \epsilon$ and $\E_{X \sim F}(|X|^2) \leq \epsilon \gamma ^2$, where $\mathds{I}$ denotes the indicator function. An important distribution in this class is $q_o(X) \triangleq q(|X|) \triangleq (1- \epsilon) \delta_0(|X|) + \epsilon \delta_\gamma(|X|)$, where $\delta_\gamma(|X|) \triangleq \delta_0(|X|- \gamma)$. Note that this distribution is independent of the phase and in addition to a point mass at zero has another point mass at $\gamma$. Finally, define $\mathcal{F}_{\epsilon}  \triangleq \{ \ F \ | \  \E_{X \sim F}(\mathds{I}(|X|\neq 0)) \leq \epsilon \}$. \\

\subsection{Performance criteria}\label{sec:criteria}
We compare c-LASSO with r-LASSO in both the noise-free and noisy
measurements cases. For each scenario, we
 define a specific measure to compare the performance of the two
 algorithms.
 \subsubsection{Noise-free measurements}\label{ssec:phasetran}
 Consider the problem of recovering $s_o$ drawn from $s_{o,i} \overset{\rm i.i.d.}{\sim} (1- \rho \delta) \delta_0(|s_{o,i}|) + \rho \delta G(s_{o,i})$, from a set of noise free measurements $y = A s_o$. Let $\mathcal{A}_{\alpha}$ be a sparse recovery algorithm with free parameter $\alpha$. For instance $\mathcal{A}$ may be the c-LASSO algorithm and the free parameter of the algorithm is the regularization argument $\lambda$. Given $(y,A)$, $\mathcal{A}_{\alpha}$ returns an estimate $\hat{x}^{\mathcal{A}_{\alpha}}$ of $s_o$.
 Suppose that in the noise free case, as $N \rightarrow \infty$, the performance of $\mathcal{A}_{\alpha}$ exhibits a sharp phase transition, i.e., for every value of $\delta$, there exists ${\rho}^{\mathcal{A}_{\alpha}}(\delta)$, below which $\lim_{N \rightarrow \infty} \|\hat{x}^{\mathcal{A}_{\alpha}}- s_o\|^2/N \rightarrow 0$ almost surely, while for $\rho > \rho^{\mathcal{A}_{\alpha}}(\delta)$, $\mathcal{A}_{\alpha}$ fails and $\lim_{N \rightarrow \infty} \|\hat{x}^{\mathcal{A}_{\alpha}}- s_o\|^2/N \nrightarrow 0$. The phase transition has been studied both empirically and theoretically for many sparse recovery algorithms \cite{MaDo09sp, NeTr08, DoTa09, DoTa09b, BlDa09, BlDaCS09, DoDrTsSt08}. The phase transition curve $\rho^{\mathcal{A}_{\alpha}}(\delta)$ specifies the fundamental exact recovery limit of algorithm $\mathcal{A}_{\alpha}$. \\
 The free parameter $\alpha$ can strongly affect the performance of the sparse recovery algorithm \cite{MaDo09sp}.
 Therefore, optimal tuning of this parameter is essential in practical applications. One approach is to tune the parameter for the highest phase transition \cite{MaDo09sp},\footnote{In this paper, we consider algorithms whose phase transitions do not depend on the distribution $G$ of non-zero coefficients. Otherwise, one could use the maximin framework introduced in \cite{MaDo09sp}.} i.e.,
 \[
 \rho^{\mathcal{A}} (\delta) \triangleq  \sup_{\alpha} \rho^{\mathcal{A}_{\alpha}}(\delta).
 \]
In other words, $ \rho^{\mathcal{A}}$ is the best performance
$\mathcal{A}_{\alpha}$ provides in the exact sparse
signal recovery problem, if we know how to tune the algorithm
properly. Based on this framework, we say algorithm $\mathcal{A}$ \textit{outperforms} $\mathcal{B}$ at a given $\delta$, if and only if $\rho^{\mathcal{A}}(\delta)>\rho^{\mathcal{B}}(\delta)$.

 \subsubsection{Noisy measurements}\label{ssec:noisyphasetran}
Consider the problem of recovering $s_o$ distributed according to $s_{o,i} \overset{\rm i.i.d.}{\sim} (1- \rho \delta) \delta_0(|s_{o,i}|) + \rho \delta G(s_{o,i})$, from a set of noisy linear observations $y = A s_o + w$, where $w_i \overset{\rm i.i.d.}{\sim} CN(0,\sigma^2)$. In the presence of measurement noise exact recovery is not possible.
 Therefore, tuning the parameter for the highest phase transition curve does not necessarily provide the optimal performance.
 In this section, we explain the \textit{optimal noise sensitivity tuning} introduced in \cite{DoMaMoNSPT}.
 Consider the $\ell_2$-norm as a measure for the reconstruction error and assume that $\frac{ \|\hat{x}^{\mathcal{A}_{\alpha}}- s_o \|_2^2}{N} \rightarrow {\rm MSE}(\rho, \delta, \alpha, \sigma, G)$ almost surely. Define the \textit{noise sensitivity} of the algorithm $\mathcal{A}_{\alpha}$ as
 \begin{eqnarray}\label{eq:ns_max}
{\rm NS}(\rho,\delta, \alpha) \triangleq \sup_{\sigma>0} \sup_G  \frac{{\rm MSE}(\rho, \delta, \alpha, \sigma, G)}{\sigma^2},
 \end{eqnarray}
 where $\alpha$ denotes the tuning parameter of the algorithm $\mathcal{A}_{\alpha}$. If the noise sensitivity is large, then the measurement noise may severely
 degrade the final reconstruction. In \eqref{eq:ns_max} we search for the distribution that induces the maximum reconstruction error to the algorithm. This ensures that for other signal distributions the reconstruction error is smaller. By tuning $\alpha$, we may obtain better estimate of $s_o$. Therefore, we tune the parameter $\alpha$ to obtain the lowest noise sensitivity, i.e.,
 \[
{\rm NS}(\rho,\delta) \triangleq \inf_{\alpha} {\rm NS}(\rho,\delta, \alpha).
 \]
 Based on this framework, we say that algorithm $\mathcal{A}$ outperforms $\mathcal{B}$ at a given $\delta$ and $\rho$ if and only if ${\rm NS}^{\mathcal{A}}(\delta,\rho)< {\rm NS}^{\mathcal{B}}(\delta,\rho)$.

    \begin{figure}
\begin{center}
  \includegraphics[width=6.2cm, height= 6cm]{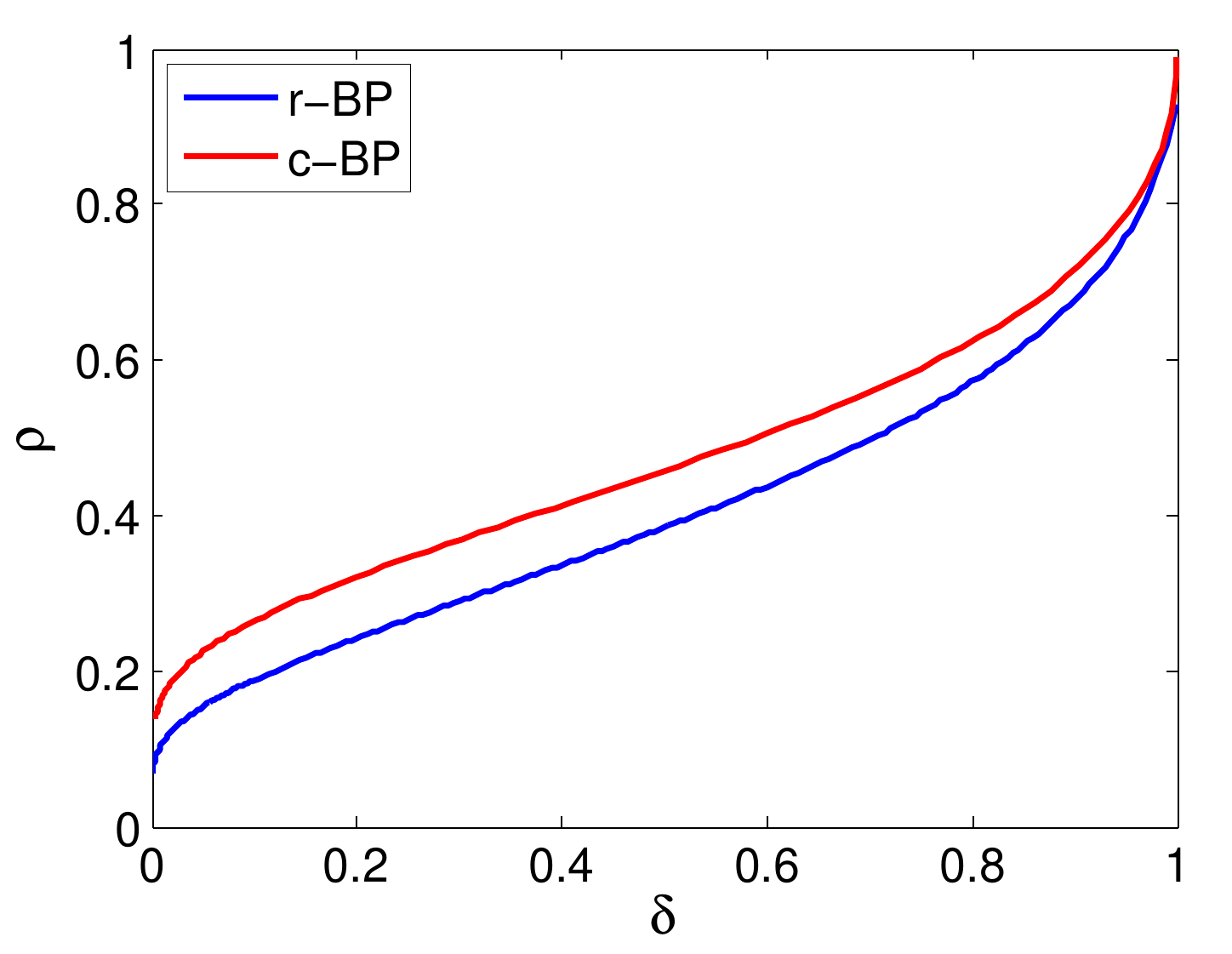}\\
  \caption{Comparison of the phase transition curve of the r-BP and c-BP. When all the non-zero real and imaginary parts of the signal are grouped, the phase transition of c-BP
  outperforms that of r-BP.}
   \label{fig:pt_realcomplex}
  \end{center}
\end{figure}

 \subsection{Contributions} \label{sec:contribution}
In this paper, we first develop the \textit{complex approximate message passing} (CAMP) algorithm that is a simple and fast converging iterative method for
 solving c-LASSO.  We extend the \textit{state evolution} (SE), introduced recently as a framework for
 accurate asymptotic predictions of the AMP performance, to CAMP.\footnote{Note that SE has been proved to be accurate only for the case of Gaussian measurement matrices \cite{BaMo10, BaMo11}. But, extensive simulations have confirmed its accuracy for a large class of random measurement matrices \cite{DoMaMo09, DoMaMoNSPT}. The results of our paper are also provably correct for complex Gaussian measurement matrices. But, our simulations confirm that they hold for broader set of matrices. } We will then use the connection between CAMP and c-LASSO to provide an accurate asymptotic analysis
 of the c-LASSO problem. We aim to characterize the phase transition curve (noise-free measurements) and noise sensitivity (noisy measurements) of c-LASSO and CAMP when the real and imaginary parts are paired, i.e., they are both zero or non-zero simultaneously. Both criteria have been extensively studied for the real signals (and hence for the r-LASSO) \cite{DoMaMo09, DoMaMoNSPT}.
The results of our predictions are summarized
 in Figures \ref{fig:pt_realcomplex}, \ref{fig:ns_complex}, and \ref{fig:ns_realcomplex}. Figure \ref{fig:pt_realcomplex} compares the phase transition curve of c-BP and CAMP with the phase transition curve of r-BP. As we expected c-BP outperforms r-BP since it exploits the connection between the real and imaginary parts. If $\rho_{SE}(\delta)$ denotes the phase transition curve, then we also prove that $\rho_{SE}(\delta) \sim \frac{1}{\log(1/2\delta)}$ as $\delta \rightarrow 0$. Comparing this with $\rho^R_{SE}(\delta) \sim \frac{1}{2\log(1/\delta)}$ for the r-LASSO \cite{DoTa09}, we conclude that
 \[
 \lim_{\delta \rightarrow 0} \frac{\rho_{SE}(\delta)}{\rho^R_{SE}(\delta)} = 2.
 \]
 This means that, in the very high undersampling regime the c-LASSO can recover signals that are two times more dense than the signals that are recovered by r-LASSO. 
  Figure \ref{fig:ns_complex} exhibits the noise sensitivity of c-LASSO and CAMP.  We prove in Section \ref{sec:noisyformal} that, as the sparsity approaches the phase transition curve, the noise sensitivity grows up to infinity. Finally, Figure \ref{fig:ns_realcomplex} compares the contour plots of the noise sensitivity of c-LASSO with those of the r-LASSO. For the fixed value noise sensitivity, the level set of the c-LASSO is higher than that of r-LASSO. It is worth noting that the same comparisons hold between CAMP and AMP, as we will clarify in Section \ref{sec:connection}. \\

\begin{figure}
\begin{center}
  \includegraphics[width=6.9cm, height = 6.5cm]{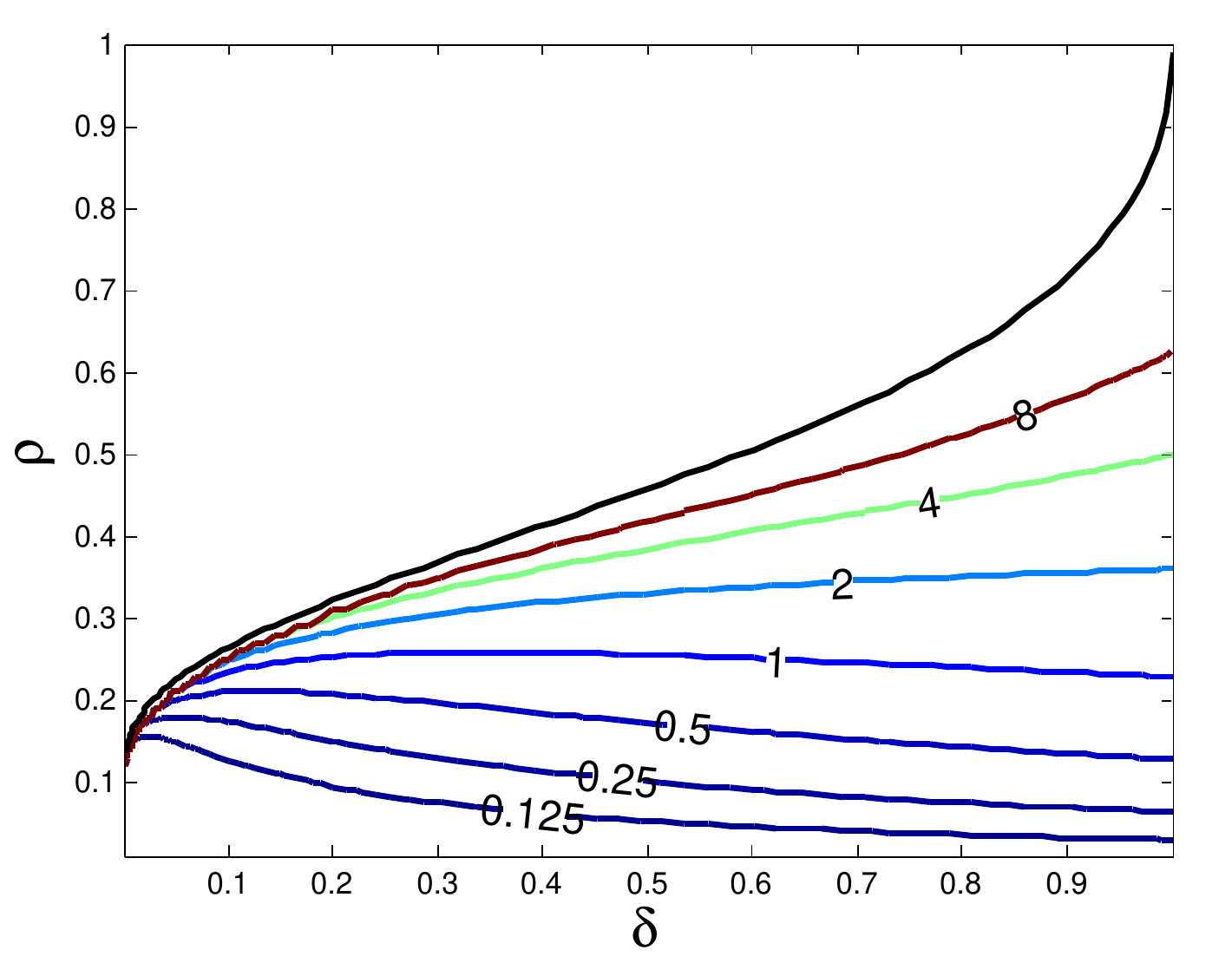}\\
  \caption{Contour lines of noise sensitivity in the $(\delta,\rho)$ plane. The black curve is the phase transition curve at which the noise sensitivity is infinite. The colored lines
  display the level sets of  ${\rm NS}(\rho,\delta) = 0.125, 0.25,0.5, 1,2,4,8$.}
   \label{fig:ns_complex}
  \end{center}
\end{figure}

  \begin{figure}
\begin{center}
  \includegraphics[width=6.5cm, height = 6.5cm]{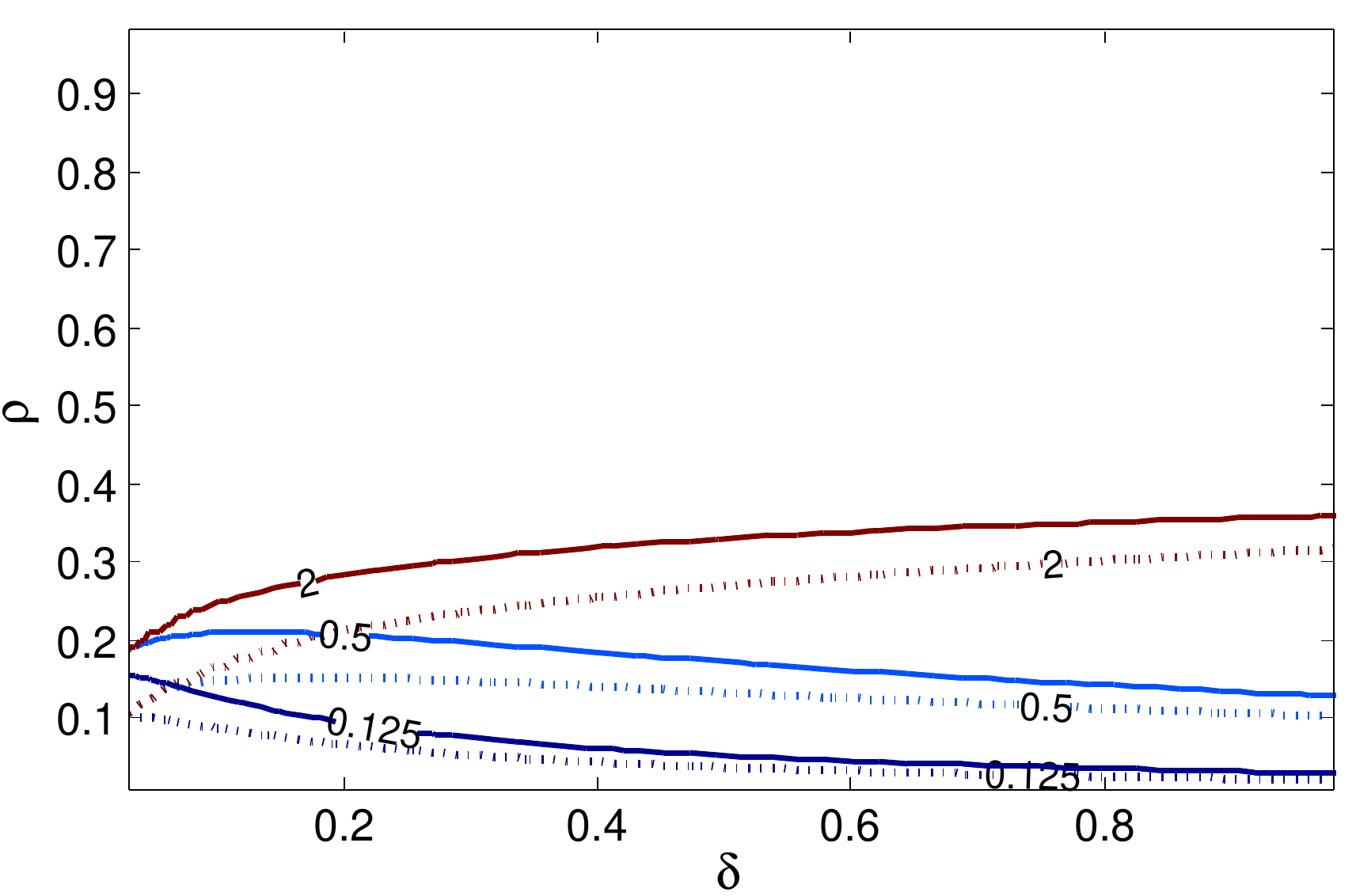}\\
  \caption{Comparison of the noise sensitivity of r-LASSO with the noise sensitivity of c-LASSO. The colored solid lines present the level sets of the ${\rm NS}(\rho,\delta) = 0.125, 0.5,2$  for the c-LASSO, and the colored dotted lines display the same level sets for the r-LASSO.}
   \label{fig:ns_realcomplex}
  \end{center}
\end{figure}

 \subsection{Related work} \label{sec:relatedwork}
 The state evolution framework used in this paper was first introduced in \cite{DoMaMo09}. Deriving the phase transition and noise sensitivity of the LASSO for
 real-valued signals and real-valued measurements from SE is due to \cite{DoMaMoNSPT}; see \cite{MalekiThesis} for more comprehensive discussion. Finally, the derivation of AMP from the full sum-product message passing is due to \cite{DoMaMo2011}. Our main contribution in this paper is to extend these results to the complex setting. Not only is the analysis of the state evolution more challenging in this setting, but it also provides new insights on the performance of c-LASSO that have not been available. For instance, the noise sensitivity of c-LASSO has not previously been determined. \\

 The recovery of sparse complex signals is a special case of group-sparsity or block-sparsity, where all the groups are non-overlapping and have size $2$. According to the group sparsity assumption, the non-zero elements of the signal tend to occur in groups or clusters. One of the algorithms used in this context is the group-LASSO \cite{BakinThesis, YuLi06}. Consider a signal $s_o \in \mathds{R}^N$. Partition the indices of $s_o$ into $m$ groups $g_1, \ldots, g_m$. The group-LASSO algorithm minimizes the following cost function:
 \begin{eqnarray}
\min_x \frac{1}{2}\|y-Ax\|_2^2 + \sum_{i=1}^m  \lambda_i  \|x_{g_i}\|_2,
 \end{eqnarray}
 where the $\lambda_i$'s are regularization parameters. 
 
 The group-Lasso algorithm has been extensively studied in the literature \cite{huang2009benefit,  peng2010regularized, duarte2011performance, van2010theoretical, chen2006theoretical, eldar2010block, lvgroup, stojnic2009reconstruction, stojnic2010ell, stojnic2009block, ji2009multi, BakinThesis, MeVaBu08, YuLi06, BachConsistency, NaRi08, LoPoTsGe09, MaCeWi05}. We briefly review several papers and emphasize the differences from our work.  \cite{BachConsistency} analyzes the consistency of the group LASSO estimator in the presence of  noise. Fixing the signal $s_o$, it provides conditions under which the group LASSO is consistent as $n \rightarrow \infty$. \cite{ObWaJo11, NaRi08}  consider a weak notion of consistency, i.e., exact support recovery. However, \cite{ObWaJo11} proves that in the setting we are interested in, i.e., $k/n = \rho$  and $n/N= \delta$, even exact support recovery is not possible. When noise is present, our goal is neither exact recovery nor exact support recovery. Instead,
 we characterize the mean square error (MSE) of the reconstruction.  This criterion has been considered in \cite{huang2009benefit, LoPoTsGe09}. Although the results of \cite{huang2009benefit, LoPoTsGe09} show qualitatively the benefit of group sparsity, they do not characterize the difference quantitatively. In fact, loose constants in both the error bound and the number of samples do not permit accurate performance comparison. In our analysis, no loose constant is involved, and we provide very accurate characterization of the mean square error. \\

 Group-sparsity and group-LASSO are also of interest in the sparse recovery community.
 For example, the analysis carried out in \cite{ duarte2011performance, lvgroup, eldar2010block} are based on ``coherence". These results provide sufficient conditions with again loose constants as discussed above.
 The work of \cite{stojnic2009reconstruction, stojnic2010ell,stojnic2009block} addresses this issue by an accurate analysis of the algorithm in the
noiseless setting $\sigma =0$.  They provide a very accurate estimate of the phase transition
curve for the group-LASSO. However, SE provides a more flexible
framework to analyze c-LASSO than the analysis of
\cite{stojnic2009block}, and it provides
more information than just the phase transition curve. For instance, it points to the least favorable distribution of the input and noise sensitivity of c-LASSO.  \\

The Bayesian approach that assumes a hidden Markov model for the signal has been also explored for the recovery of group sparse signals \cite{SchSel11, Schniter10}. It has been shown that AMP combined with an expectation maximization algorithm (for estimating the parameters of the distribution) leads to promising results in practice \cite{ViSc12}. Kamilov et al. \cite{KamRang12} have taken the first step towards a theoretical understanding of such algorithms. However, the complete understanding of the expectation maximization employed in such methods is not available yet. Furthermore, the success of such algorithms seem to be dependent on the match between the assumed and actual prior distribution. Such dependencies have not been theoretically analyzed yet. In this paper we assume that the distribution of non-zero coefficients is not known beforehand and characterize the performance of c-LASSO for the least favorable distribution.  \\

While writing this paper we were made aware that in an independent work Donoho, Johnstone, and Montanari are extending the SE framework to the general setting of group sparsity \cite{DoMo11}.
Their work considers the state evolution framework for the group-LASSO problem and will include the generalization of the analysis provided in this paper to the
case where the variables tend to cluster in groups of size $B$. \\

Both complex signals and group-sparse signals are
special cases of model-based CS
\cite{RichModelBasedCS}. By introducing more structured models for the signal,
\cite{RichModelBasedCS} proves that the number of
measurements needed are proportional to the ``complexity'' of the
model rather than the sparsity level \cite{JaMa11}. The results in model-based CS also suffer from loose constants in both the number of
measurements and the mean square
error bounds.  \\

 Finally, from an algorithmic point of view, several papers have considered solving the c-LASSO problem using first-order algorithms \cite{BeFr08, FiNoWr07}.\footnote{First-order methods are iterative algorithms that use either the gradient or the subgradient of the function at the previous iterations to update their estimates.} The deterministic framework that measures the convergence of an algorithm on the problem instance that yields the slowest convergence rate, is not an appropriate measure  of the convergence rate for the compressed sensing problems \cite{MaBa11}. Therefore, \cite{MaBa11} considers the average convergence rate for iterative algorithms. In that setting, AMP is the only first order algorithm that provably
 achieves linear convergence to date. Similarly, the CAMP algorithm, introduced in this paper, provides the first, first-order c-LASSO solver that provides a linear average convergence rate.

 \subsection{Organization of the paper} \label{sec:organization}
We introduce the CAMP algorithm in Section \ref{sec:CAMP}. We then explain the state evolution equations that characterizes the evolution of the mean square error through the
iterations of the CAMP algorithm in Section \ref{sec:formal}, and we analyze the important properties of the SE equations. We then discuss the
connection between our calculations and the solution of LASSO in Section \ref{sec:connection}. We confirm our results via Monte Carlo simulations in Section \ref{sec:simulation}.

\section{Complex Approximate Message Passing} \label{sec:CAMP}
The high computational complexity of interior point methods for
solving large scale convex optimization problems has spurred the development
of first-order methods for solving the LASSO problem. See
\cite{MaBa11} and the references therein for a description of some of
these algorithms. One of the most successful algorithms for
CS problems is the AMP algorithm introduced in
\cite{DoMaMo09}. In this section, we use the approach introduced in
\cite{DoMaMo2011} to derive the approximate message passing
algorithm for the c-LASSO problem that we term \textit{Complex
Approximate Message Passing} (CAMP).

Let $s_1, s_2, \ldots, s_N$ be $N$ random variables with the
following distribution:
\begin{eqnarray}\label{eq:distlaplacegauss}
p(s_1,s_2, \ldots, s_N) = \frac{1}{Z(\beta)} {\rm e}^{- {\beta \lambda}\|s\|_1 -\frac{\beta}{2} \|y-As\|_2^2},
\end{eqnarray}
where $\beta$ is a constant and $Z(\beta) \triangleq \int_s {\rm e}^{- \beta \lambda \|s\|_1 -\frac{\beta}{2} \|y-As\|_2^2 }  ds$.
As $\beta \rightarrow \infty$, the mass of this distribution concentrates around the solution of the LASSO.
Therefore, one way to find the solution of LASSO is to marginalize this distribution. However, calculating the marginal
distribution is an NP-complete problem. The sum-product message passing algorithm provides a successful heuristic
for approximating the marginal distribution. As $N \rightarrow \infty$ and $\beta \rightarrow \infty$ the iterations of the sum-product message
passing algorithm are simplified to (\cite{DoMaMo2011} or Chapter 5 of \cite{MalekiThesis})
\begin{eqnarray} \label{eq:fullMP}
x_{\ell \rightarrow a}^{t+1} &=& \eta\Big(\sum_{b \neq a}A^*_{b\ell}z^t_{b \rightarrow \ell}; \tau_t\Big), \nonumber\\
z_{a \rightarrow \ell}^t &=& y_a - \sum_{j \neq \ell} A_{aj} x^t_{j \rightarrow a},
\end{eqnarray}
where $\eta (u+ iv; \lambda) \triangleq \left(u+iv-
\frac{\lambda(u+iv)}{\sqrt{u^2+v^2}}\right)_+\mathds{I}_{\{u^2+v^2>
\lambda^2\}}$ is the proximity operator of the complex
$\ell_1$-norm and is called complex soft thresholding. See Appendix \ref{appsec:prox} for further
information regarding this function. $\tau_t$ is the threshold parameter at time
$t$. The choice of this parameter will be discussed in Section
\ref{sec:stateevolution}. The per-iteration computational complexity
of this algorithm is high, since $2 n N$ messages ${x^t_{\ell
\rightarrow a}}$ and $z^t_{a \rightarrow \ell}$ are updated. Therefore,
following \cite{DoMaMo2011} we assume that there exist $\Delta x_{\ell \rightarrow a}^t,  \Delta z_{\ell \rightarrow a}^t = O({1}/{\sqrt{N}})$ such that
\begin{eqnarray} \label{eq:mp_assumption}
x_{\ell \rightarrow a}^t &=& x_\ell^t+ \Delta x_{\ell \rightarrow a}^t+ O(1/N), \nonumber \\
z_{a \rightarrow \ell}^t &=& z_a^t+ \Delta z_{\ell \rightarrow a}^t + O(1/N).
\end{eqnarray}
 Here, the $O(\cdot)$ errors are uniform in the choice of the edges $\ell \rightarrow a$ and $a \rightarrow \ell$. In other words we assume that $x_{\ell \rightarrow a}^t$ is independent of $a$ and $z_{a \rightarrow \ell}^t$ is independent of $\ell$ except for an error of order $1/\sqrt{N}$. For further discussion of this assumption and its validation, see \cite{DoMaMo2011} or Chapter 5 of \cite{MalekiThesis}. Let $\eta^I$ and $\eta^R$ be the imaginary and real parts of the complex soft thresholding function. Furthermore, define $\frac{\partial \eta^R}{\partial x}$ and $\frac{\partial \eta^R}{\partial y}$ as the partial derivatives of $\eta^R$ with respect to the real and
imaginary parts of the input respectively.  $\frac{\partial \eta^I}{\partial x}$, and $\frac{\partial \eta^I}{\partial y}$ are defined similarly. The following theorem shows how one can simplify the message passing as $N \rightarrow \infty$.

\begin{propo} \label{prop:ampderivation}
Suppose that \eqref{eq:mp_assumption} holds for every iteration of the message passing algorithm specified in \eqref{eq:fullMP}. Then $x_\ell^t$
and $z_a^t$ satisfy the following equations:
\begin{eqnarray}\label{eq:camp}
x_\ell^{t+1} &=& \eta\Big(x_\ell^t+ \sum_b A^*_{b\ell} z_{b}^t ; \tau_t\Big), \nonumber \\
z_a^{t+1} & = & y_a - \sum_j A_{aj} x_j^{t+1}\nonumber\\
&-& \!\! \! \!\sum_j A_{aj} \left(\frac{\partial \eta^R}{\partial x}\Big(x_j^{t} + \sum_b A^*_{bj} z_b^{t}\Big) \right) \mathcal{R}(A^*_{aj}z_a^t) \nonumber\\
&-& \!\! \! \!\sum_j A_{aj} \left(\frac{\partial \eta^R}{\partial y} \Big(x_j^{t} + \sum_b A^*_{bj} z_b^{t}\Big) \right)  \mathcal{I}(A^*_{aj}z_a^t) \nonumber\\
&-&\!\!\! \! \! \! i \sum_j A_{aj} \left(\frac{\partial \eta^I}{\partial x}\Big(x_j^{t} + \sum_b A^*_{bj} z_b^{t}\Big) \right)  \mathcal{R}(A^*_{aj}z_a^t) \nonumber\\
&-&\! \!\! \! \! \! i \sum_j A_{aj} \left(\frac{\partial \eta^I}{\partial
y} \Big(x_j^{t} + \sum_b A^*_{bj} z_b^{t}\Big) \right)
\mathcal{I}(A^*_{aj}z_a^t).
\end{eqnarray}
\end{propo}
See Appendix \ref{appsec:CAMP} for the proof. According to Proposition  \ref{prop:ampderivation} and \eqref{eq:mp_assumption}, for large values of $N$,
the messages $x_{\ell \rightarrow a}^t$ and $z_{a \rightarrow \ell}^t$ are close to $x_\ell^t$ and $z_a^t$ in \eqref{eq:camp}. Therefore, we define the CAMP algorithm as the iterative method that
starts from $x^0 =0$ and $z^0 =y$ and uses the iterations specified in \eqref{eq:camp}. It is important to note that Proposition \ref{prop:ampderivation} does not provide any information
on either the performance of the CAMP algorithm or the connection between CAMP and c-LASSO, since message passing
is a heuristic algorithm and does not necessarily converge to the correct marginal distribution of \eqref{eq:distlaplacegauss}.

\section{Formal analysis of CAMP and c-LASSO}\label{sec:formal}
In this section, we explain the state evolution (SE) framework that predicts the performance of the CAMP and c-LASSO in
the asymptotic settings. We then use this framework to analyze the
phase transition and noise sensitivity of the CAMP and c-LASSO. The formal 
connection between state evolution and CAMP/c-LASSO is discussed in Section \ref{sec:connection}.

\subsection{State evolution}\label{sec:stateevolution}
We now conduct an asymptotic analysis of the CAMP algorithm. As we
confirm in Section \ref{sec:connection}, the asymptotic performance
of the algorithm is tracked through a few variables, called the
state variables. The state of the algorithm is the 5-tuple $
\mathbf{s} = (m;\delta, \rho, \sigma, G)$, where $G$ corresponds to
the distribution of the non-zero elements of the sparse vector
$s_o$, $\sigma$ is the standard deviation of the measurement noise,
and $m$ is the asymptotic normalized mean square error. The
threshold parameter (threshold policy) of CAMP in its most general form
could be a function of the state of the algorithm
$\tau(\mathbf{s})$. Define  ${\rm npi}(m; \sigma,\delta)\triangleq
\sigma^2 + \frac{m}{\delta}$. The {\em mean square error (MSE) map} is defined as
\begin{eqnarray*}
{\Psi( \mathbf{s}, \tau(\mathbf{s}) )  \triangleq} \hspace{7.6cm} \nonumber \\
\E|\eta(X+ \sqrt{ \npi(m,\sigma,\delta)} Z_1 +i \sqrt{ \npi(m,\sigma,\delta)} Z_2; \tau( \mathbf{s})) - X|^2, \ \ \ \ \   
\end{eqnarray*}
where $Z_1, Z_2 \sim N(0,1/2)$ and $X \sim (1- \rho \delta) \delta_0(|x|)+ \rho \delta G(x)$ are independent random variables. Note that $G$ is a probability distribution on $\mathds{C}$.
In the rest of this paper, we consider the thresholding policy $\tau(\mathbf{s}) = \tau \sqrt{{\rm npi}(m,\sigma,\delta)}$, where the constant $\tau$ is yet to be tuned according to the schemes introduced in Sections \ref{ssec:phasetran} and \ref{ssec:noisyphasetran}. When we use this thresholding policy we may equivalently write $\Psi(\mathbf{s}, \tau(\mathbf{s}))$ as $\Psi(\mathbf{s}, \tau)$. This thresholding policy is the same as the thresholding policy introduced in \cite{DoMaMoNSPT,DoMaMo09}. When the parameters $\delta, \rho,\sigma, \tau$ and $G$ are clear from the context, we
denote the MSE map by $\Psi(m)$. SE is the evolution of $m$ (starting from $t=0$ and $m_0= E(|X|^2)$) by the rule
\begin{eqnarray}\label{eq:stateevo1}
\lefteqn{m_{t+1} =  \Psi(m_t) }  \nonumber \\
  &\triangleq&  \! \! \! \!  \E\Big|\eta\left( \! X+ \!  \sqrt{ \npi^t} Z_1 +\! i \sqrt{ \npi^t} Z_2; \tau \sqrt{ \npi^t} \right) \! - \!  X \Big|^2 \!  \!  \! ,
\end{eqnarray}
where $\npi^t \triangleq \npi(m_t,\sigma,\delta)$.
As will be described in Section \ref{sec:connection}, this equation tracks the normalized MSE of the CAMP algorithm in the asymptotic setting $n, N \rightarrow \infty$ and $n/N \rightarrow \delta$. In other
words, if $m_t$ is the MSE of the CAMP algorithm at iteration $t$, the $m_{t+1}$, calculated by \eqref{eq:stateevo1}, is the MSE of CAMP at iteration $t+1$.
\begin{definition}
Let $\Psi$ be almost everywhere differentiable. $m^*$ is called a fixed point of $\Psi$ if and only if $\Psi(m^*) = m^*$. Furthermore, a fixed point is called stable if $\left.\frac{d\Psi(m)}{dm} \right|_{m = m^*} <1$, and unstable if $\left.\frac{d\Psi(m)}{dm} \right|_{m = m^*} >1$.
\end{definition}
It is clear that if $m^*$ is the unique stable fixed point of the $\Psi$ function, then $m_t \rightarrow m^*$ as $t \rightarrow \infty$. Also, if all the fixed points of $\Psi$ are unstable, then $m_t \rightarrow \infty$ as $t \rightarrow \infty$. Define $\mu \triangleq |X|$ and $\theta \triangleq \measuredangle X$. Let $G(\mu,\theta)$ denote the probability density function of $X$ and define $G(\mu) \triangleq  \int G(\mu, \theta) d\theta$ as the marginal distribution of $\mu$. The next lemma shows that in order to analyze the state evolution function we only need to consider the amplitude distribution. This substantially simplifies our analysis of SE in the next sections.  
\begin{lemma}\label{lem:phaseindep}
The MSE map does not depend on the phase distribution of the input signal, i.e.,
\[
\Psi(m,\delta, \rho, \sigma, G(\mu, \theta), \tau) = \Psi(m,\delta,
\rho, \sigma, G(\mu), \tau).
\]
\end{lemma}
\noindent See Appendix \ref{appsec:proofphaseindep} for the proof.

\subsection{Noise-free signal recovery}\label{sec:noisefreeformal}

Consider the noise free setting with $\sigma =0$. Suppose that SE predicts the MSE of CAMP in the asymptotic setting (we will make this rigorous in Section \ref{sec:connection}). As mentioned in Section \ref{ssec:phasetran}, in order to characterize the performance of CAMP in the noiseless setting, we first derive its phase transition curve and then optimize over $\tau$ to obtain the highest phase transition CAMP can achieve. Fix
all the state variables except for $m$, and $\rho$. The evolution of
$m$, discriminates the following two regions for $\rho$:
\begin{itemize}
\item[] Region I: The values of $\rho$ for which $\Psi(m) < m$ for every
$m>0$;
\item[] Region II: The complement of Region I.
\end{itemize}

Since $0$ is necessarily a fixed point of the $\Psi$ function, in Region I $m_t \rightarrow 0$ as $t \rightarrow \infty$. The following lemma shows that in Region II $m=0$ is an unstable fixed point and therefore starting from $m_0 \neq 0$, $m_t \nrightarrow 0$.
\begin{lemma}
Let $\sigma =0$. If $\rho$ is in Region II, then $\Psi$ has an unstable fixed point at zero.
\end{lemma}
\begin{proof}
We prove in Lemma \ref{lem:concavity} that $\Psi(m)$ is a concave
function of $m$.  Therefore, $\rho$ is in Region II if and only if
$\left. \frac{d\Psi(m)}{dm} \right|_{m =0} >1$. This in turn
indicates that $0$ is an unstable fixed point.
\end{proof}

It is also easy to confirm that Region I
is of the form $[0, \rho_{SE}(\delta, G, \tau))$. As we will see in Section \ref{sec:connection}, $\rho_{SE}(\delta, G, \tau)$ determines the phase transition curve of the CAMP algorithm. According to Lemma \ref{lem:phaseindep}, the MSE map does not depend on the phase distribution of the non-zero elements. The following proposition shows that in fact $\rho_{SE}$ is independent of $G$ even though the $\Psi$ function depends on $G(\mu)$.

\begin{propo}\label{prop:distributionindep}
$\rho_{SE}(\delta, G, \tau)$ is independent of the distribution $G$.
\end{propo}
\begin{proof}
According to Lemma \ref{lem:concavity} in Appendix \ref{appsec:ptformula}, $\Psi$ is concave. Therefore, it has a stable fixed point at zero if and only if its derivative
at zero is less than $1$. It is also straightforward (from Appendix \ref{appsec:ptformula}) to show that
\[
\left.\frac{d\Psi}{dm}\right|_{m=0} = \frac{\rho \delta(1+ \tau^2)}{\delta} + \frac{1- \rho\delta}{\delta}\E|\eta(Z_1+ iZ_2; \tau)|^2.
\]
Setting this derivative to $1$, it is clear that the phase transition value of $\rho$ is independent of $G$.
\end{proof}

According to Proposition
\ref{prop:distributionindep} the only parameters that affect
$\rho_{SE}$ are $\delta$ and the free parameter $\tau$. Fixing $\delta$, we tune $\tau$ such that the algorithm achieves
its highest phase transition for a certain number of measurements, i.e.,
\begin{eqnarray*}
\rho_{SE}(\delta) \triangleq \sup_{\tau} \rho_{SE}(\delta; \tau).
\end{eqnarray*}
Using SE we can calculate the optimal value of $\tau$ and $\rho_{SE}(\delta)$.

\begin{thm}\label{thm:phasetranscurve}
$\rho_{SE}(\delta)$ and $\delta$ satisfy the following implicit relations:
\begin{eqnarray*}
\rho_{SE}(\delta) &=& \frac{ \chi_1(\tau) }{(1+ \tau^2)\chi_1(\tau)- \tau \chi_2(\tau)}, \\
\delta &=& \frac{4(1+ \tau^2)\chi_1(\tau)- 4 \tau
\chi_2(\tau) }{-2 \tau + 4 \chi_2(\tau)},
\end{eqnarray*}
for $\tau \in [0, \infty)$. Here, $\chi_1(\tau) \triangleq \int_{\omega \geq \tau}\omega(\tau- \omega){\rm e}^{- \omega^2} d\omega $ and $\chi_2(\tau) \triangleq \int_{\omega > \tau} \omega (\omega - \tau)^2 {\rm e}^{- \omega ^2}$. 
\end{thm}
\noindent See Appendix \ref{appsec:ptformula} for the proof. Figure \ref{fig:pt_realcomplex} displays this phase transition curve that is derived from the SE framework and compares it with the phase transition of r-BP algorithm. As will be described later, $\rho_{SE}(\delta)$ corresponds to the phase transition of c-LASSO. Hence the difference between $\rho_{SE} (\delta)$ and phase transition curve of r-LASSO is the benefit of grouping the real and imaginary parts. 

It is also interesting to compare the $\rho_{SE}(\delta)$ (which as we see later predicts the performance of c-LASSO) with the phase transition of r-LASSO in high undersampling regime $\delta \rightarrow 0$. The implicit
formulation above enables us to calculate the asymptotic performance
of the phase transition as $\delta \rightarrow 0$.

\begin{thm}\label{thm:ptasumptote}
$\rho_{SE}(\delta)$ follows the asymptotic behavior
\[
\rho_{SE}(\delta) \sim \frac{1}{ \log\left(\frac{1}{2 \delta} \right) },  \ \  {\rm as} \ \ \ \delta \rightarrow 0.
\]
\end{thm}
\noindent See Appendix \ref{appsec:laplace} for the proof. As mentioned above, this theorem shows that as $\delta \rightarrow 0$
the phase transition of c-BP and CAMP is two times that of the r-LASSO, which is given by $\rho^R_{SE} \sim 1/ (2 \log(1/\delta))$ \cite{DoTa09}. This improvement is due to the grouping of real and imaginary parts of the signal.

\subsection{Noise sensitivity} \label{sec:noisyformal}
In this section we characterize the noise sensitivity of SE. To achieve this goal, we first discuss the risk of the complex soft thresholding function.
The properties of this risk play an important role in the discussion
of the noise sensitivity of SE in Section
\ref{sec:ns_se}.
\subsubsection{Risk of soft thresholding}\label{sec:softrisk}
Define the risk of the soft thresholding function as
\[
r(\mu, \tau) \triangleq \E |\eta(\mu {\rm e}^{i \theta}+ Z_1 + i Z_2; \tau)- X |^2,
\]
where $\mu \in [0, \infty)$, $\theta \in [0, 2\pi)$, and the expected value is with respect to the two independent
random variables $Z_1, Z_2 \sim N(0,1/2)$. It is important to note
that according to Lemma \ref{lem:phaseindep}, the risk function is
independent of $\theta$. The following lemma characterizes two
important properties of this risk function:
\begin{lemma}\label{lem:softthreshincconcave}
$r(\mu, \tau)$ is an increasing function of $\mu$ and a concave function in terms of $\mu^2$.
\end{lemma}
\noindent See Appendix \ref{app:pf_concavityrisk} for the proof of this lemma. We define the minimax risk of the soft thresholding function as
 \[
 M^{\flat}(\epsilon) \triangleq \inf_{\tau >0} \sup_{q \in \mathcal{F}_{\epsilon}} \E |\eta(X+ Z_1 + i Z_2; \tau)- X |^2,
\]
where $q$ is the probability density function of $X$, and the expected value is with respect to $X$, $Z_1$ and $Z_2$. \\

Note that $q \in \mathcal{F}_{\epsilon}$ implies that $q$ has a point mass of $1-\epsilon$ at zero; see Section \ref{sec:notation} for more information. In the next section we show a connection between this minimax risk and the noise sensitivity of the SE. Therefore, it is important to characterize $M^{\flat}(\epsilon)$.
\begin{propo}\label{prop:minimaxsoft}
The minimax risk of the soft thresholding function satisfies
\begin{equation}\label{eq:minimaxrisk_soft}
M^{\flat}(\epsilon) = \inf_{\tau}  \ 2 (1-\epsilon)  \int_{w= \tau}^{\infty} w(w- \tau)^2 {\rm e}^{- w^2}dw + \epsilon(1+ \tau^2).
\end{equation}
\end{propo}
See Appendix \ref{app:proofprop} for the proof. It is important to note that the quantities in \eqref{eq:minimaxrisk_soft} can be easily calculated in terms of the density and distribution function of a normal random variable. Therefore, a simple computer program may accurately calculate the value of $M^{\flat}(\epsilon)$ for any $\epsilon$.  

The proof provided for Proposition \ref{prop:minimaxsoft} also proves the following proposition. We will discuss the importance of this result for compressed sensing problems in the next section.

\begin{propo}\label{prop:leastfavor}
The maximum of the risk function, $ \max_{q \in \mathcal{F}_{\epsilon, \gamma}} \E |\eta(X+ Z_1 + i Z_2; \tau)- X |^2$, 
 is achieved on $q(X) = (1- \epsilon) \delta_0(|X|) + \epsilon \delta_{\gamma}(|X|)$. 
\end{propo} 
First, note that the maximizing distribution (or least favorable distribution) is independent of the threshold parameter. Second, note that the maximizing distribution is not unique since we have already proved that the phase distribution does not affect the risk function. 

\subsubsection{Noise sensitivity of state evolution}\label{sec:ns_se}
 As mentioned in Section \ref{sec:stateevolution}, in the presence of measurement noise, SE is given by
\begin{eqnarray*}
m_{t+1} &=& \Psi(m_t) \nonumber \\
 &=& \E|\eta(X+ \sqrt{{\rm npi}} Z_1+ i \sqrt{{\rm npi}}  Z_2 ; \tau \sqrt{{\rm npi}}) - X|^2,
\end{eqnarray*}
where ${\rm npi} ={\sigma^2+\frac{m_t}{\delta}}$. As mentioned above, $m_t$ characterizes the asymptotic MSE of CAMP at iteration $t$. Therefore, the final solution of the CAMP algorithm converges to one of the stable fixed points of the $\Psi$ function. The next theorem suggests that the stable fixed point is unique, and therefore no matter where the algorithm starts from it will always converge to the same MSE.

\begin{lemma}\label{lem:uniqfix}
$\Psi(m)$ has a unique stable fixed point to which the sequence of $\{m_t \}$ converges.
\end{lemma}
We call the fixed point in Lemma \ref{lem:uniqfix} ${\rm fMSE}(\sigma^2, \delta, \rho, G,
\tau)$. According to Section \ref{ssec:noisyphasetran}, we define
the minimax noise sensitivity as
\[
{\rm NS}^{SE}(\delta, \rho) \triangleq \min_{\tau} \sup_{\sigma>0} \sup_{ q \in \mathcal{F}_{\epsilon}} {\rm fMSE}(\sigma^2, \delta, \rho, G, \tau)/\sigma^2.
\]
The noise sensitivity of SE can be easily evaluated
from $ M^{\flat}(\epsilon)$. The following theorem characterizes this relation.
\begin{thm}
Let $\rho_{MSE}(\delta)$ be the value of $\rho$ satisfying $M^{\flat} (\rho \delta) = \delta$. Then, for $\rho < \rho_{MSE}$ we have
\begin{eqnarray*}
{\rm NS}^{SE}(\delta, \rho) = \frac{M^{\flat}(\delta \rho)}{1 - M^{\flat}(\delta \rho)/\delta},
\end{eqnarray*}
and for $\rho > \rho_{MSE}(\delta)$, ${\rm NS}^{SE}(\delta, \rho) = \infty$.
\end{thm}
\noindent The proof of this theorem follows along the same lines as the proof of Proposition 3.1 in \cite{DoMaMoNSPT}, and therefore we skip it for the sake of brevity.  The contour lines of this noise sensitivity function are displayed in Figure \ref{fig:ns_complex}. 

Similar arguments as those presented in Proposition 3.1 in \cite{DoMaMoNSPT} combined with Proposition \ref{prop:leastfavor} prove the following.
\begin{propo}
The maximum of the formal MSE, $\max_{q \in \mathcal{F}_{\epsilon, \gamma}}{\rm fMSE}(\sigma^2, \delta, \rho, G, \tau)$
is achieved by $q = (1- \epsilon) \delta_0(|X|)+\epsilon \delta_{\gamma}(|X|)$, independent of $\sigma$ and $\tau$. 
\end{propo}
Again we emphasize that the maximizing or least favorable distribution is not unique. Note that the least favorable distribution provides a simple approach for designing and setting the parameters of CS systems \cite{AnMaBa12}: We design the system such that it performs well on the least favorable distribution, and it is then guaranteed that the system will perform as well (or in many cases better) on all other input distributions.

As a final remark we note that $\rho_{MSE}(\delta)$ equals $\rho_{SE}(\delta)$ as proved next.

\begin{propo} For every $\delta \in [0,1]$ we have
\[
\rho_{MSE}(\delta) = \rho_{SE}(\delta).
\]
\end{propo}
\begin{proof}
The proof is a simple comparison of the formulas. We first know that $\rho_{MSE}$ is derived from the following equation
\[
\min_{\tau} 2(1- \rho \delta) \int_{\omega > \tau} \omega(\omega- \tau)^2 {\rm e}^{-\omega^2} d\omega + \rho \delta (1+ \tau^2) = \delta.
\]
On the other hand, since $\Psi(m)$ is a concave function of $m$, $\rho_{SE}(\delta, \tau)$ is derived from $\left. \frac{d\Psi(m)}{dm}\right|_{m =0} =1$. This derivative is
equal to
\[
\left. \frac{d\Psi(m)}{dm}\right|_{m =0} \! \! = \frac{2(1- \rho \delta)}{\delta} \int_{\omega > \tau} \! \! \omega(\omega- \tau)^2 {\rm e}^{-\omega^2} d\omega + \frac{\rho \delta}{\delta} (1+ \tau^2).
\]
Also, $\rho_{SE}(\delta)= \sup_{\tau}\rho_{SE}(\tau,\delta)$. However, in order to obtain the highest $\rho$ we should minimize the above expression over $\tau$.
Therefore, both $\rho_{SE}(\delta)$ and $\rho_{MSE}(\delta)$ satisfy the same equations and thus are exactly equal.
\end{proof}

\subsection{Connection between the state evolution, CAMP, and c-LASSO }\label{sec:connection}
 There is a strong connection between the SE framework, the CAMP algorithm, and c-LASSO.
 Recently, \cite{BaMo10} proved that SE predicts the asymptotic performance of the AMP algorithm when the measurement matrix is i.i.d. Gaussian. The result also holds for complex Gaussian matrices and complex input vectors. As in \cite{DoMaMo09},
we conjecture that the SE predictions are correct for a ``large" class of random
matrices. We show evidence of this claim in Section
\ref{sec:simulation}. Here, for the sake of completeness, we quote
the result of \cite{BaMo10} in the complex setting. Let $\gamma :
\mathds{C}^2\rightarrow \mathds{R}$ be a pseudo-Lipschitz
function.\footnote{$\gamma: \mathds{C}^2 \rightarrow \mathds{R}$ is
pseudo-Lipschitz if and only if $|\psi(x)-\psi(y)| \leq L(1+ \|x\|_2
+ \|y\|_2)\|x-y\|_2$.} To make the presentation clear we consider a simplified version of Definition 1 in \cite{BaMo11}.  
\begin{definition}
A sequence of instances $\{s_o(N), A(N), w(N)\}$, indexed by the ambient dimension $N$, is called a converging sequence if the following conditions hold:
\begin{itemize}
\item[-] The elements of $s_o(N) \in \mathds{R}^N$ are i.i.d. drawn from $(1-\rho \delta)\delta_0(|s_{o,i}|) + \rho \delta G(s_{o,i})$. 
\item[-] The elements of $w(N) \in \mathds{R}^n$ ($n = \delta N$) are i.i.d. drawn from $N(0, \sigma_w^2)$.
\item[-] The elements of $A(N) \in \mathds{R}^{n \times N}$ are i.i.d. drawn from a complex Gaussian distribution.
\end{itemize}
\end{definition}

\begin{thm}\label{thm:dynamic message}
Consider a converging sequence $\{s_o(N), A(N), w(N)\}$. Let $x^t(N)$ be the estimate of the CAMP algorithm at iteration $t$. For any pseudo Lipschitz
function $\gamma: \mathds{C}^2 \rightarrow \mathds{R}$ we have
\begin{eqnarray}
\lefteqn{\lim_{N \rightarrow \infty} \frac{1}{N} \sum_{i=1}^{N} \gamma(x_i^t, s_{o,i})} \nonumber \\
&=& \E \gamma\left(\eta\left(X+ \sqrt{\npi^t}Z_1+ i \sqrt{\npi^t}Z_2; \tau \sqrt{\npi^t}\right), X\right)\nonumber
\end{eqnarray}
almost surely, where $Z_1 +i Z_2 \sim CN(0,1)$  and $X \sim (1-\rho \delta)\delta_0(|s_{o,i}|) + \rho \delta G(s_{o,i})$ are independent complex random variables. Also, $\npi^t \triangleq \sigma^2+ m_t/\delta$, where $m_t$ satisfies \eqref{eq:stateevo1}. 
\end{thm}
The proof of this theorem is similar to the proof of Theorem 1 in \cite{BaMo10} and hence is skipped here.\\

It is also straightforward to extend the result of \cite{DoMaMoNSPT} and \cite{BaMo11} on the connection of message passing algorithms and LASSO to the complex setting. For a given value of $\tau$ suppose that the fixed point of the state evolution is denoted by $m^*$. Define $\lambda(\tau)$ as
\begin{equation}\label{eq:taulambdacalibration}
\lambda(\tau)\! 
 \triangleq\! \tau \sqrt{m^*} \left(1 - \frac{1}{2\delta} \E \! \left(\frac{\partial  \eta^R}{\partial x} + \frac{\partial \eta^I}{\partial y}\right)\right),
\end{equation}
where
\begin{eqnarray}
\frac{\partial  \eta^R}{\partial x} &\triangleq& \frac{\partial  \eta^R}{\partial x} \! \left(X+ \sqrt{m^*} Z_1\! + \! i\sqrt{m^*} Z_2; \tau \sqrt{m^*} \right), \nonumber \\
 \frac{\partial \eta^I}{\partial y} &\triangleq& \frac{\partial \eta^I}{\partial y} \! \left(X+ \sqrt{m^*} Z_1 +i\sqrt{m^*} Z_2; \tau \sqrt{m^*}\right), \nonumber
\end{eqnarray}
and $\E$ is with respect to independent random variables $Z_1 +i Z_2 \sim CN(0,1)$  and $X \sim (1-\rho \delta)\delta_0(|s_{o,i}|) + \rho \delta G(s_{o,i})$. The following theorem establishes the connection between the solution of LASSO and the state evolution equation. 

\begin{thm}\label{thm:dynamicLASSO}
Consider a converging sequence $\{s_o(N), A(N), w(N)\}$. Let $\hat{x}^{\lambda(\tau)}(N)$ be the solution of LASSO. Then, for any pseudo Lipschitz
function $\gamma: \mathds{C}^2 \rightarrow \mathds{R}$ we have
\begin{eqnarray*}
\lefteqn{\lim_{N \rightarrow \infty} \frac{1}{N} \sum_{i=1}^{N} \gamma(\hat{x}_i^{\lambda(\tau)}(N), s_{o,i})} \nonumber \\
&=& \E \gamma\left(\eta\left(X+ \sqrt{\npi^*}Z_1+ i \sqrt{\npi^*}Z_2; \tau \sqrt{\npi^*}\right), X\right)
\end{eqnarray*}
almost surely, where $Z_1 +i Z_2 \sim CN(0,1)$  and $X \sim (1-\rho \delta)\delta_0(|s_{o,i}|) + \rho \delta G(s_{o,i})$ are independent complex random variables. $\npi^* \triangleq \sigma^2+ m^*/\delta$, where $m^*$ is the fixed point of \eqref{eq:stateevo1}. 
\end{thm}
The proof of the theorem is similar to the proof of Theorem 1.4 in \cite{BaMo11} and hence is skipped here. \\

Note that according to Theorems \ref{thm:dynamic message} and \ref{thm:dynamicLASSO}, SE predicts the dynamic of the AMP algorithm and the solution of LASSO accurately in the asymptotic settings. 

\subsection{Discussion}

\subsubsection{Convergence rate of CAMP}\label{sec:discussconv}
In this section we briefly discuss the convergence rate of the CAMP algorithm. In this respect our results are straightforward extension of the analysis in \cite{MaBa11}. But, for the sake of completeness, we mention a few highlights. Let $\{m_t \}_{t=1}^{\infty}$ be a sequence of MSE generated according to state evolution \eqref{eq:stateevo1} for $X \sim (1- \epsilon)\delta_0(|X|)+ \epsilon G(X)$, $\tau$, and $\sigma^2=0$. The following proposition provides an upper bound on $m_t$ as a function of iteration $t$. 

\begin{thm}\label{prop:convrate}
Let $\{m_t \}_{t=1}^{\infty}$ be a sequence of MSEs generated according to SE. Then
\[
m^t \leq \left( \left. \frac{d\Psi(m)}{dm}\right|_{m=0}\right)^t m_0.
\] 
\end{thm}
\begin{proof}
 Since according to Lemma \ref{lem:concavity} $\Psi(m)$ is concave, we have $\Psi(m) \leq  \left. \frac{d\Psi(m)}{dm}\right|_{m=0} m$. Hence at every iteration, $m$ is attenuated by $\left. \frac{d\Psi(m)}{dm}\right|_{m=0}$. After $t$ iterations we have $m^t \leq \left( \left. \frac{d\Psi(m)}{dm}\right|_{m=0}\right)^t m_0$.
\end{proof}

\vspace{.2cm}

According to Theorem \ref{prop:convrate} the convergence rate of CAMP is linear (in the asymptotic setting).\footnote{If the measurement matrix is not i.i.d. random CAMP does not necessarily converge at this rate. This is due to the fact that state evolution does not necessarily hold for arbitrary matrices.} In fact, due to the concavity of the $\Psi$ function, CAMP converges faster for large values of MSE $m$. As $m$ reaches zero the convergence rate decreases towards the rate predicted by this theorem. Theorem \ref{prop:convrate} provides an upper bound on the number of iterations the algorithm requires to reach to a certain accuracy. Figure \ref{fig:convrate} exhibits the value of $\left. \frac{d\Psi(m)}{dm}\right|_{m=0}$ as a function of $\rho$ and $\delta$. This figure is based on the calculations we have presented in Appendix \ref{appsec:ptformula}. Here, $\tau$ is chosen such that the CAMP algorithm achieves the same phase transition as c-BP algorithm. Note that, according to Proposition \ref{prop:convrate}, if $ \left. \frac{d\Psi(m)}{dm}\right|_{m=0} < 0.9$, then $m_{200}<7.1\times10^{-10} m_0$. 

Theorem \ref{prop:convrate} only considers the noise-free problem. But, again due to the concavity of the $\Psi$ function, the convergence of CAMP to its fixed point is even faster for noisy measurements. To see this, note that once the measurements are noisy, the fixed point of CAMP occurs at a larger value of $m$. Since $\Psi$ is concave, the derivative at this point is lower than the derivative at zero. Hence, convergence will be faster.

  \begin{figure}
\begin{center}
  \includegraphics[width=6.9cm, height = 6.5cm]{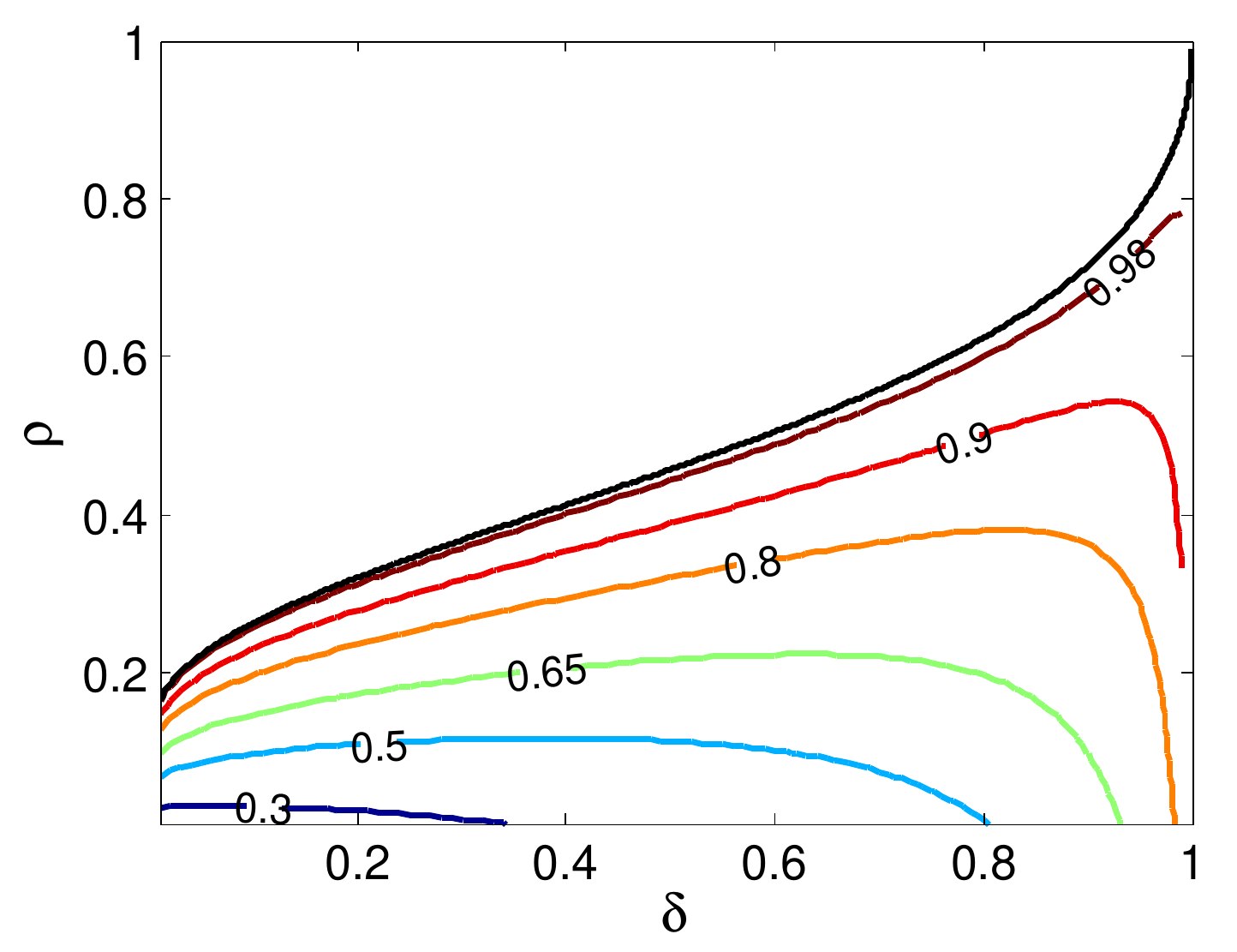}\\
  \caption{Contour lines of $ \left. \frac{d\Psi(m)}{dm}\right|_{m=0}$ as a function of $\rho$ and $\delta$. The parameter $\tau$ in the SE is set according to Theorem \ref{thm:phasetranscurve}. The black curve is the phase transition of CAMP. The colored lines
  display the level sets of  $  \left. \frac{d\Psi(m)}{dm}\right|_{m=0}= 0.3, 0.5, 0.65, 0.8, 0.9, 0.98$. Note that according to Proposition \ref{prop:convrate}, if $ \left. \frac{d\Psi(m)}{dm}\right|_{m=0} < 0.9$, then $m_{200}<7.1\times10^{-10} m_0$. }
   \label{fig:convrate}
  \end{center}
\end{figure}

\subsubsection{Extensions}
The results presented in this paper are concerned with the two most popular problems in compressed sensing, i.e., exact recovery of sparse signals and approximate recovery of sparse signals in the presence of noise. However, our framework is far more powerful and can address other compressed sensing problems as well. For instance a similar framework has been used to address the problem of recovering approximately sparse signals in the presence of noise \cite{DoJoMaMo11}. For the sake of brevity we have not provided such an analysis in the current paper.  However, the properties we proved in Lemmas \ref{lem:phaseindep}, \ref{lem:softthreshincconcave}, and Proposition \ref{prop:minimaxsoft} enable a straightforward extension of our analysis to such cases as well. 

Furthermore, the framework we developed here provides a way for recovering sparse complex-valued signals when the distribution of non-zero elements is known. This area has been studied in \cite{Schniter10, ViSc12}.

\section{Simulations}\label{sec:simulation}
As explained in Section \ref{sec:connection}, our theoretical results show that, if the elements of the matrix are i.i.d. Gaussian, then SE
predicts the performance of the CAMP and c-LASSO algorithms accurately. However, in this section we will show evidence that suggests the theoretical
framework is applicable to a wider class of measurement matrices. We then investigate the dependence of the empirical phase transition on the input distribution for medium
problem sizes.

\subsection{Measurement matrix simulations} \label{sec:setup}
We investigate the effect of the measurement matrix distribution on the performance of CAMP and c-LASSO in two different cases. First, we consider the case
where the measurements are noise-free. We postpone a discussion of measurement noise to Section \ref{sec:matuniv_noisy}.
\begin{figure}\vspace{-.1cm}
\begin{center}
\subfigure{  \includegraphics[width=6cm, height =
6.2cm]{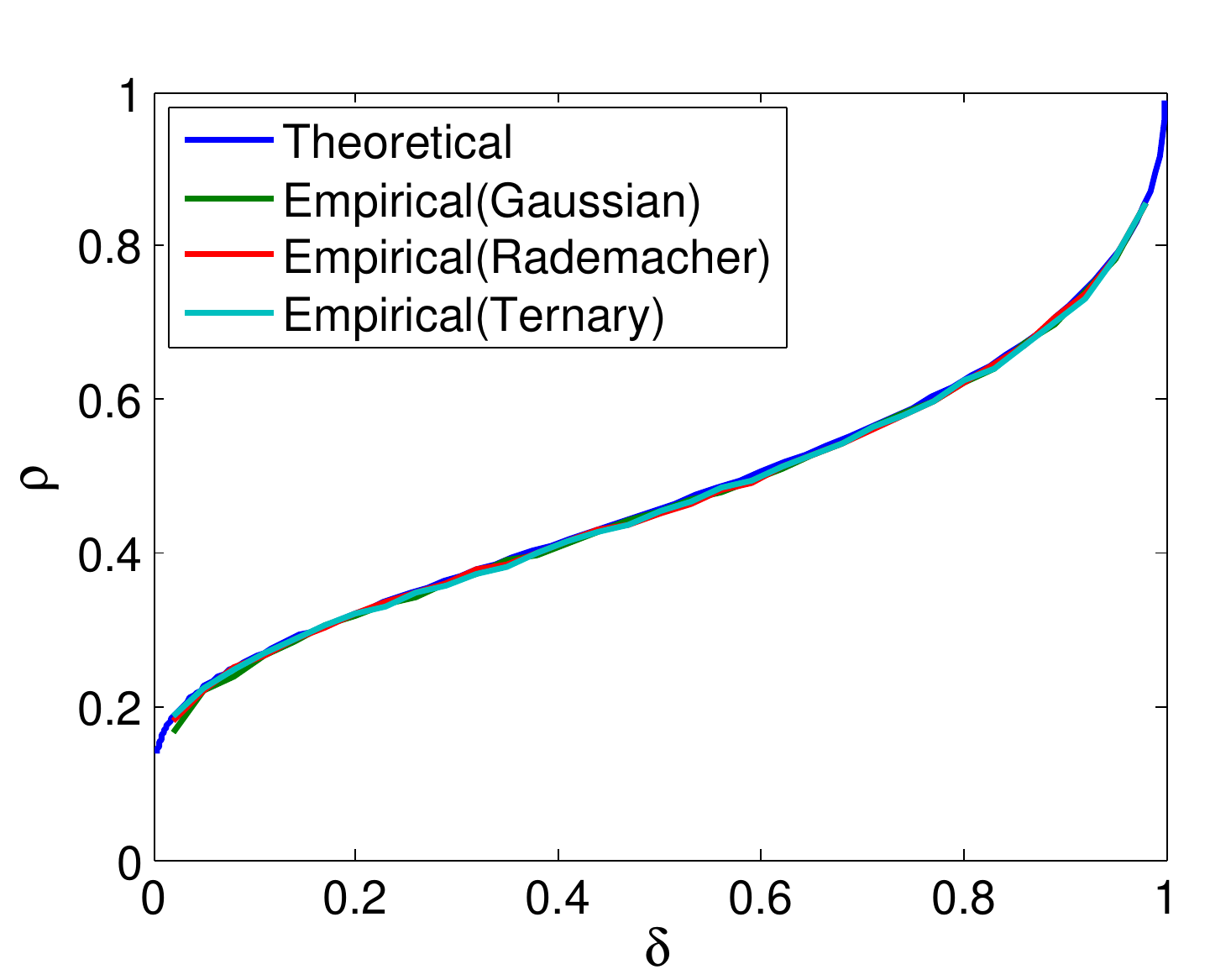}}
\end{center}
\begin{center}
 \subfigure{  \includegraphics[width=6cm,
height=6cm]{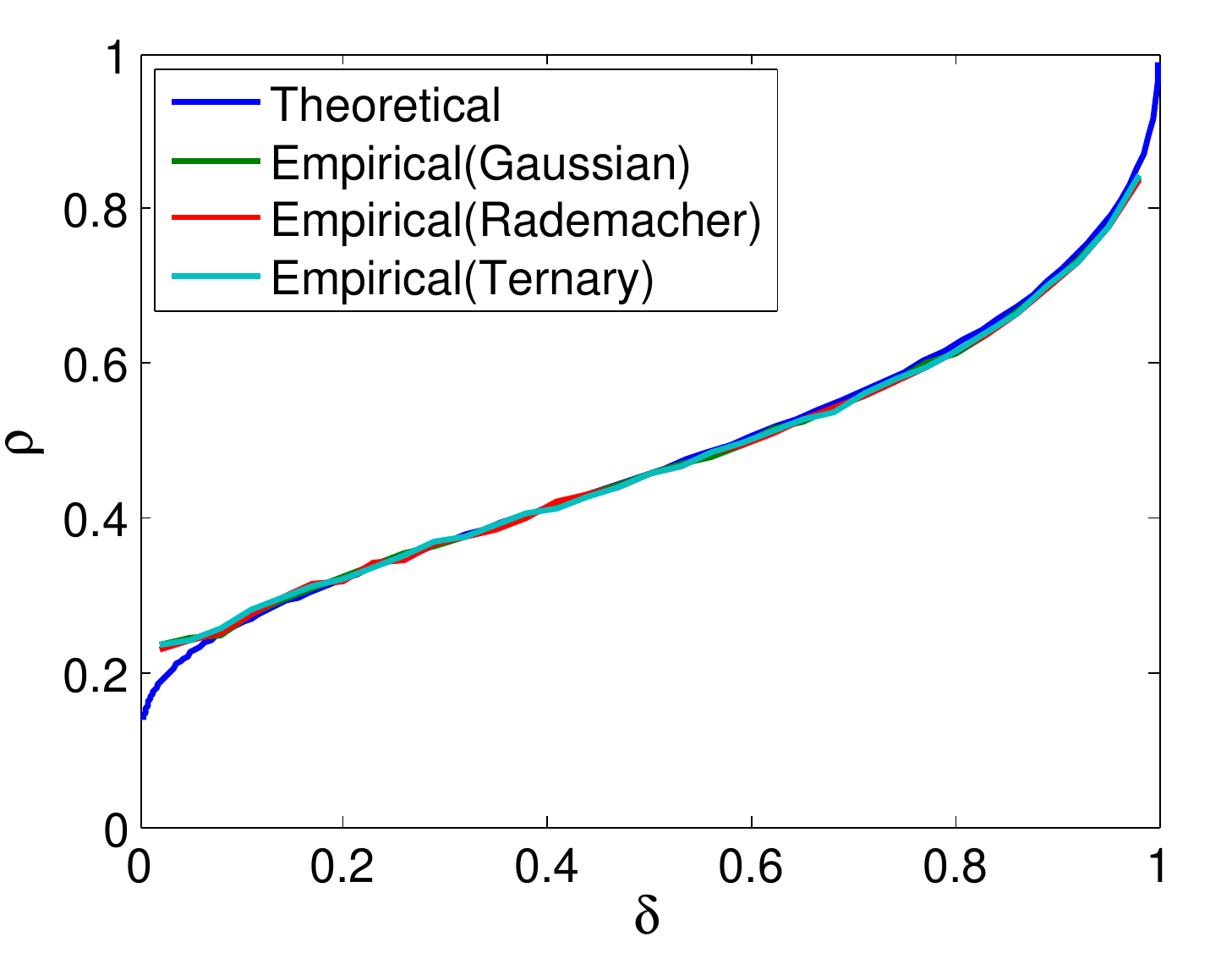}} \vspace{0cm}
\end{center}
\caption{Comparison of $\rho_{SE}(\delta)$ with the empirical phase
transition of c-LASSO \cite{YaZh11} (top) and CAMP (bottom). There
is a close match between the theoretical prediction and the
empirical results from Monte Carlo simulations.}
\label{fig:pt_matuniv}
\end{figure}

 \begin{table}[ht]
\caption{Ensembles considered for the measurement matrix $A$ in the matrix universality ensemble experiments.} 
\vspace{2mm}
\centering 
\begin{tabular}{|l|l|}
  \hline
  Name & Specification \\
  \hline
  Gaussian  & i.i.d. elements with $CN(0,1/n)$ \\
  \hline
  Rademacher  & i.i.d. elements with real and imaginary parts distributed \\ 
  &according to $\frac{1}{2}\delta_{-\sqrt{\frac{1}{2n}}}(x)+\frac{1}{2} \delta_{\sqrt{\frac{1}{2n}}}(x)$ \\
   \hline
 Ternary   & i.i.d. elements with real and imaginary parts distributed \\ 
 &according to $\frac{1}{3}\delta_{-\sqrt{\frac{3}{2n}}}(x)+\frac{1}{3}\delta_0(x) +\frac{1}{3} \delta_{\sqrt{\frac{3}{2n}}}(x)$\\
   \hline
   \end{tabular}
\label{table:matrixfdist}
\end{table}

\subsubsection{Noise-free measurements}
Suppose that the measurements are noise-free. Our goal is to empirically measure the phase transition curves of the c-LASSO and CAMP
on the measurement matrices provided in Table \ref{table:matrixfdist}.
To characterize the phase transition of an algorithm, we do the following:
\begin{itemize}
\item[-] We consider 33 equispaced values of $\delta$ between $0$ and $1$.
\item[-] For each value of $\delta$, we calculate $\rho_{SE}(\delta)$ from the theoretical framework and then consider 41 equispaced values of
$\rho$ in $[\rho_{SE}(\delta)-0.2, \rho_{SE}(\delta)+0.2]$.
\item[-] We fix $N=1000$, and for any value of $\rho$ and $\delta$, we calculate $n =\lfloor \delta N \rfloor$ and $k = \lfloor  \rho \delta N \rfloor$.
\item[-] We draw $M=20$ independent random matrices from one of the distributions described in Table \ref{table:matrixfdist} and for each matrix we construct a random input vector $s_o$ with one of the 
distributions described in Table \ref{table:coeffdist}. We then form $y = As_o$ and recover $s_o$ from
$y,A$ by either c-BP or CAMP to obtain
$\hat{x}$. The matrix distributions and coefficient distributions we
consider in our simulations are specified in Tables
\ref{table:matrixfdist} and \ref{table:coeffdist}, respectively.
\item[-] For each $\delta$, $\rho$, and Monte Carlo sample $j$ we define a success variable $S_{\delta, \rho, j}= \mathds{I} \left(\frac{\| \hat{x}-x\|_2 }{\|x\|_2} < {\rm tol}\right)$ and we calculate
the success probability $\hat{p}^{S}_{\delta, \rho} = \frac{1}{M} \sum_j S_{\delta, \rho, j}$. This provides an empirical estimate of the probability of correct recovery. The value of tol in our case is set to $10^{-4}$.
\item[-] For a fixed value of $\delta$, we fit a logistic regression function to $\hat{p}^{S}(\delta, \rho)$ to obtain $p_{\delta}^S(\rho)$. Then we find the value of $\hat{\rho}_{\delta}$
for which $p_{\delta}^S(\rho) = 0.5$.
\end{itemize}
See \cite{MaDo09sp} for a more detailed discussion of this approach.
For the c-LASSO algorithm, we are reproducing the experiments of
\cite{YaZh11, yang2012phase} and, therefore, we are using one-L1 algorithm \cite{YaZh11}.
Although Figure \ref{fig:convrate} confirms that for most cases even $200$
iterations of CAMP are enough to reach convergence, since our goal is
to measure the phase transition, we consider $3000$ iterations. See Section \ref{sec:discussconv} for the discussion on the convergence rate.

Figure \ref{fig:pt_matuniv} compares the phase transition of c-LASSO and CAMP on the ensembles specified in Table \ref{table:matrixfdist} with the theoretical prediction of this paper.
In this simulation the coefficient ensemble is UP (see Table \ref{table:coeffdist}). Clearly, the empirical and theoretical phase transitions of the algorithms coincide. More importantly, we can conjecture that the choice of the measurement matrix ensemble does not affect the phase transition of these two algorithms.  We will next discuss the impact of measurement matrix when there is noise on the measurements. 
\begin{table}[ht]
\caption{Coefficient ensembles considered in coefficient ensemble experiments.} 
\vspace{2mm}
\centering 
\begin{tabular}{|l|l|}
  \hline
  Name & Specification \\
  \hline
  UP  & i.i.d. elements with amplitude $1$ and uniform phase\\
  \hline
  ZP  & i.i.d. elements with amplitude $1$ and phase zero \\
   \hline
  GA  & i.i.d. elements with standard normal real and imaginary parts \\
  \hline
  UF & i.i.d. elements with $U[0,1]$ real and imaginary parts\\
  \hline
   \end{tabular}
\label{table:coeffdist}
\end{table}

\subsubsection{Noisy measurements}\label{sec:matuniv_noisy}
In this section we aim to show that, even in the presence of noise, the matrix ensembles defined in Table \ref{table:matrixfdist} perform similarly.
 Here is the setup for our experiment:
\begin{itemize}
\item[-]  We set $\delta = 0.25$, $\rho = 0.1$, and $N =1000$.
\item[-] We choose $50$ different values of $\sigma$ in the range [0.001, 0.1].
\item [-] We choose $n \times N$ measurement matrix $A$ from one of the ensembles specified in Table \ref{table:matrixfdist}.
\item[-] We draw $k$ i.i.d. elements from UP ensemble for the $k= \lfloor \rho n \rfloor$ non-zero elements of the input $s_o$.
\item[-] We form the measurement vector $y = As_o + \sigma w$ where $w$ is the noise vector with i.i.d. elements from $CN(0,1)$.
\item[-] For CAMP, we set $\tau =2$. For c-LASSO, we use \eqref{eq:taulambdacalibration} to derive the corresponding values of $\lambda$ for $\tau =2$ in CAMP.
\item[-] We calculate the MSE $\| \hat{x}-s_o\|_2^2/N$ for each matrix ensemble and compare the results.
\end{itemize}

Figures \ref{fig:matuniv_noisylasso} and \ref{fig:matuniv_noisyCAMP}
summarize our results. The concentration of the points along the
$y=x$ line indicates that the matrix ensembles, specified in Table \ref{table:matrixfdist}, perform
similarly. The coincidence of the phase transition curves for different matrix ensembles is known as {\em universality hypothesis (conjecture)}. In order to provide a stronger evidence, we run the above
experiment with $N =4000$. The
results of this experiment are exhibited in Figures
\ref{fig:matuniv_noisylasso2} and \ref{fig:matuniv_noisyCAMP2}. It
is clear from these figures that the MSE is now more concentrated
around the $y=x$ line. Additional experiments with other parameter values exhibited the same behavior. Note that as $N$ grows, the variance of the MSE estimate becomes smaller, and the behavior of the algorithm is closer to the average performance that is predicted by the SE equation.

\begin{figure}
\begin{center}
\subfigure{  \includegraphics[width=5cm, height =5cm]{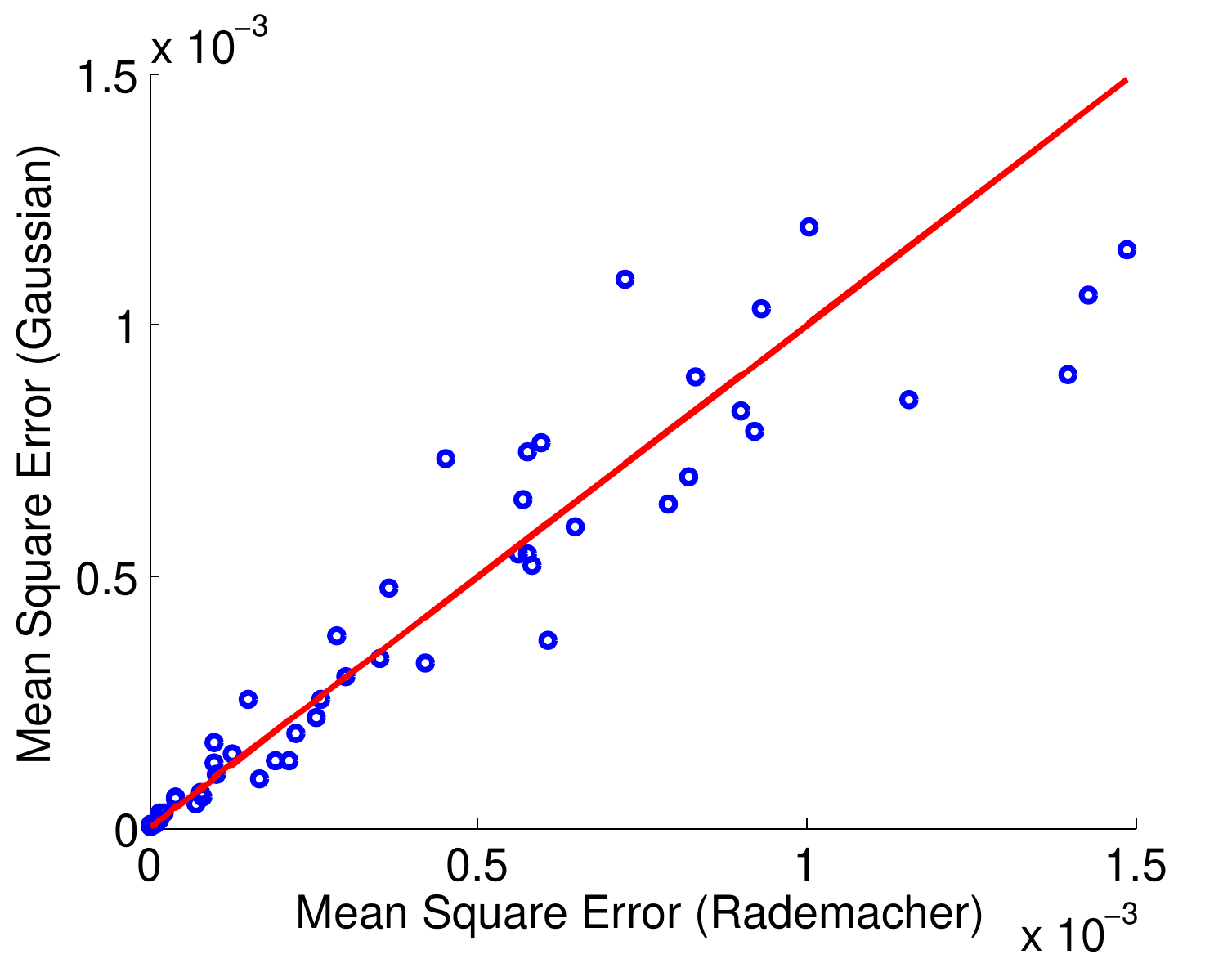}}
\end{center}
\begin{center}
\subfigure{  \includegraphics[width=5cm, height= 5cm]{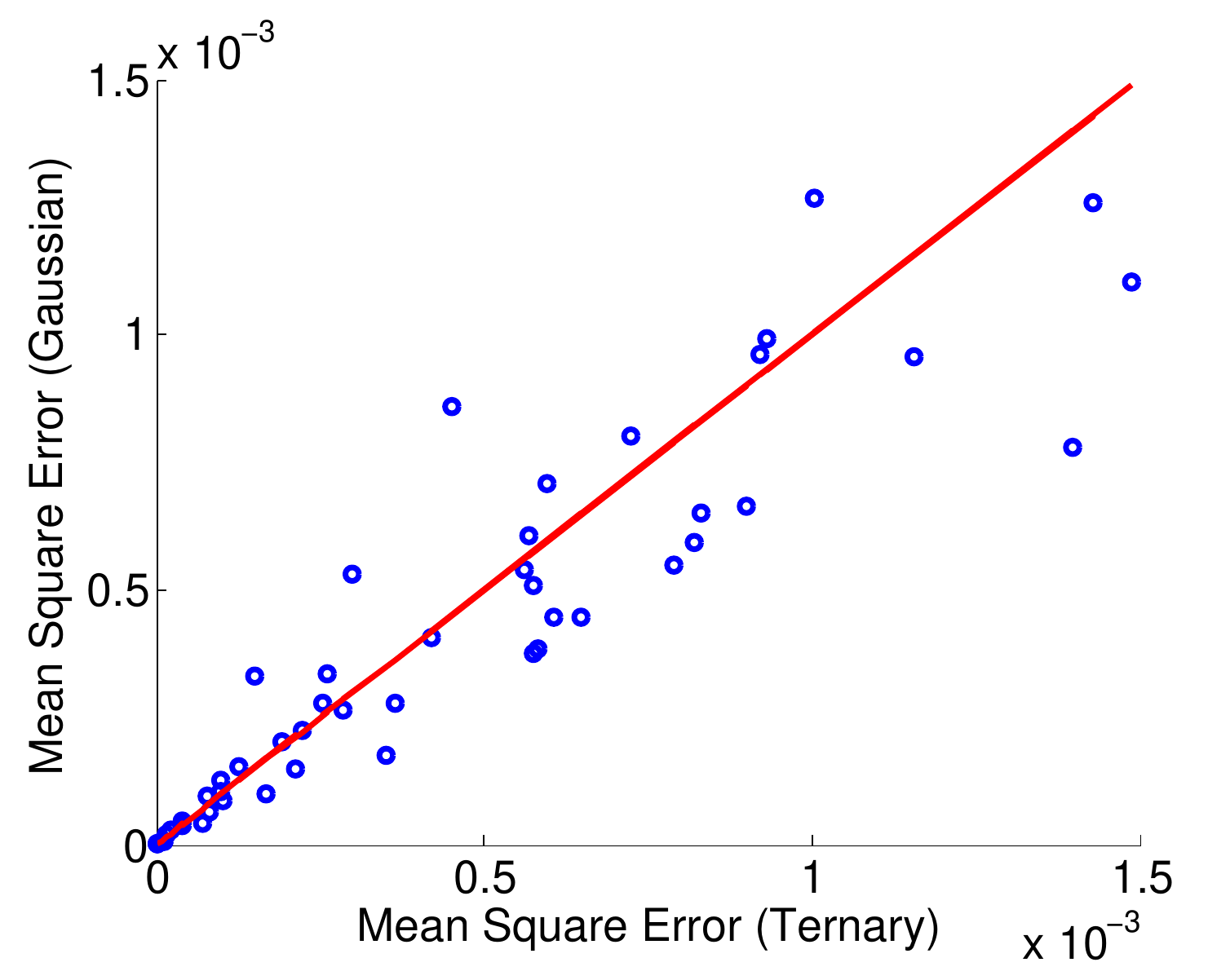}}
\end{center}
\caption{Comparison of the means square error of c-LASSO for Gaussian and
Rademacher matrix ensembles (top), and Gaussian and Ternary ensemble (bottom). The concentration of points around the $y=x$ confirms the universality hypothesis. The norms of residuals are equal to $5.9\times 10^{-4}$ and $6 \times  10^{-4}$ for the top and bottom figures, respectively. Comparison of this figure with Figure \ref{fig:matuniv_noisylasso2} confirms that as $N$ grows the points become more concentrated around $y=x$ line. } \label{fig:matuniv_noisylasso}
\end{figure}

\begin{figure}
\begin{center}
\subfigure{  \includegraphics[width=5.5cm, height = 5.5cm]{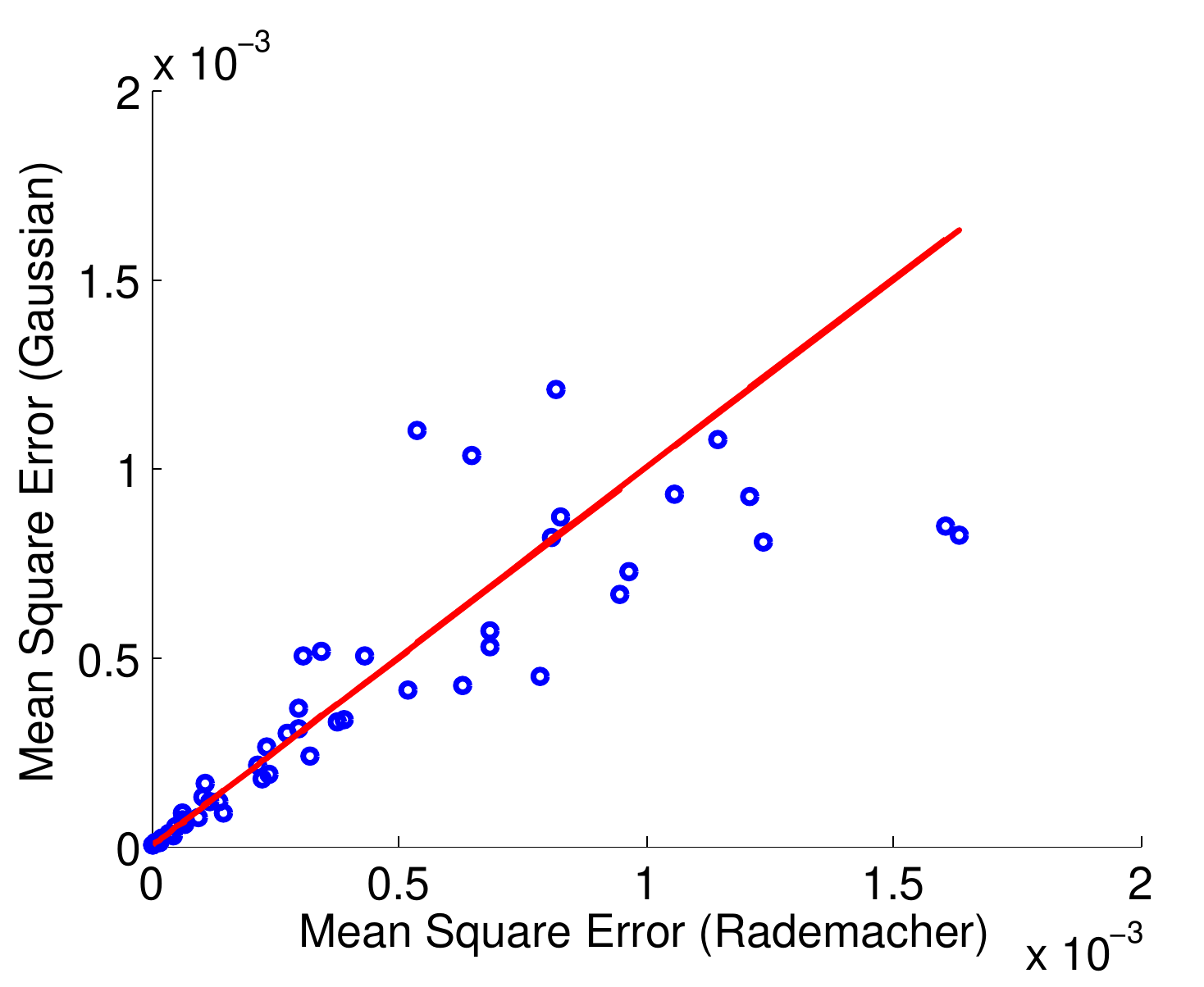}}
\end{center}
\begin{center}
\subfigure{  \includegraphics[width=5.5cm, height=5.5cm]{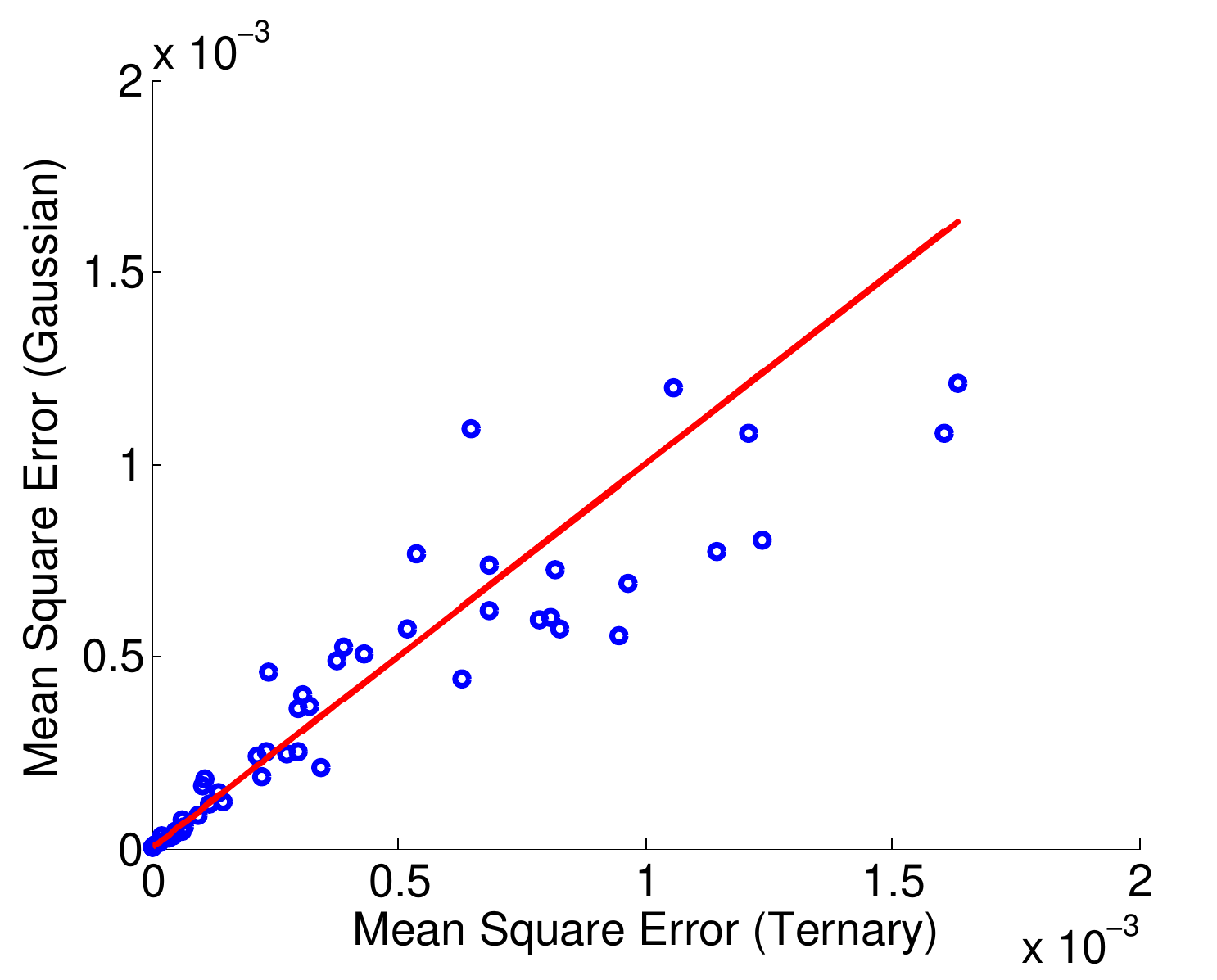}}
\end{center}
\caption{Comparison of the MSE of CAMP for Gaussian and
Rademacher matrix ensembles (top), and Gaussian and Ternary ensemble (bottom). The concentration of points around the $y = x$ line confirms the universality hypothesis. The norms of residuals are equal to $9.1\times 10^{-4}$ and $9.4 \times  10^{-4}$ for the top and bottom figures, respectively. Comparison of this figure with Figure \ref{fig:matuniv_noisyCAMP2} confirms that as $N$ grows the points become more concentrated around $y=x$ line.} \label{fig:matuniv_noisyCAMP}
\end{figure}

\begin{figure}
\begin{center}
\subfigure{  \includegraphics[width=5.5cm, height = 5.5cm]{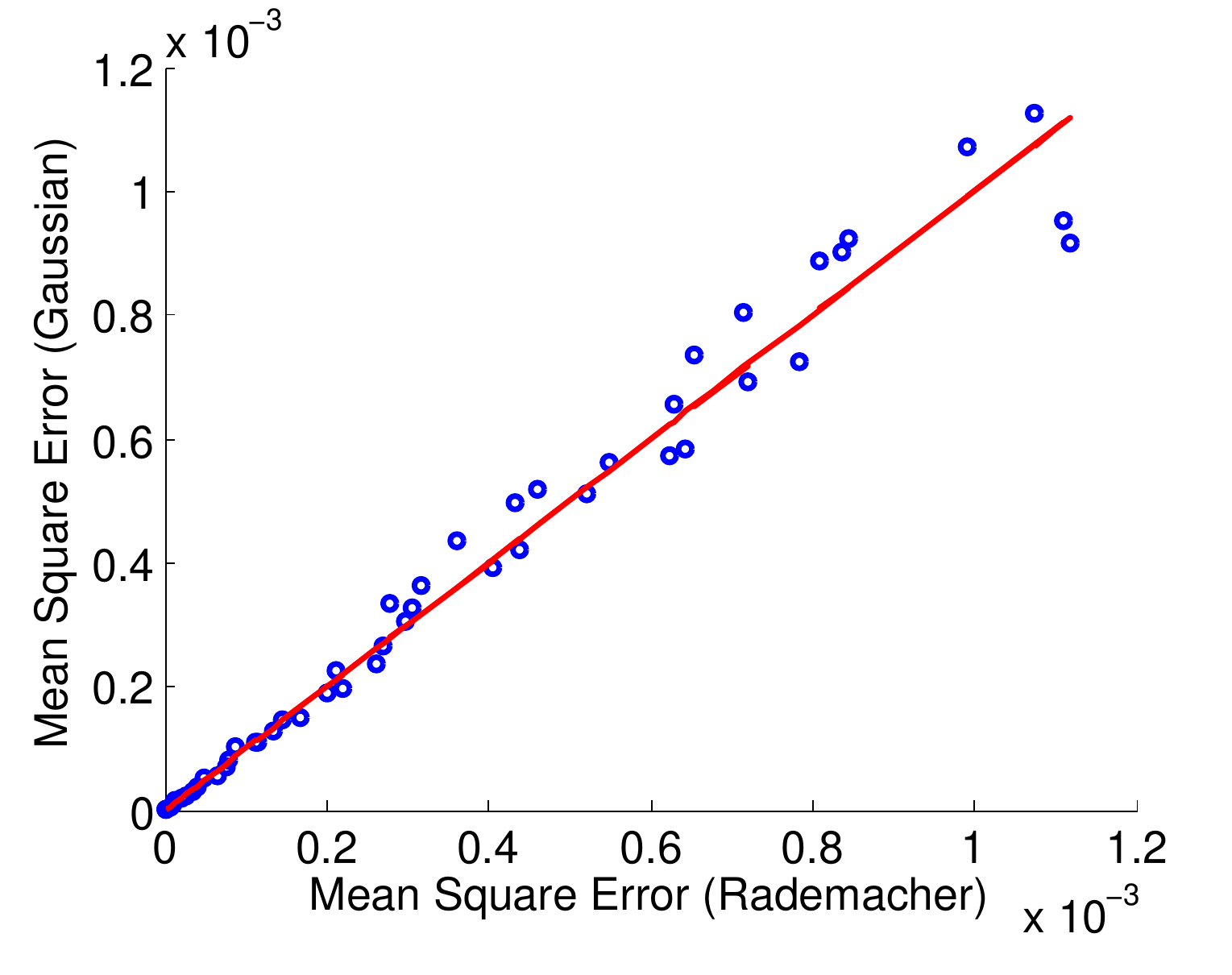}}
\end{center}
\begin{center}
\subfigure{  \includegraphics[width=5.5cm, height= 5.5cm]{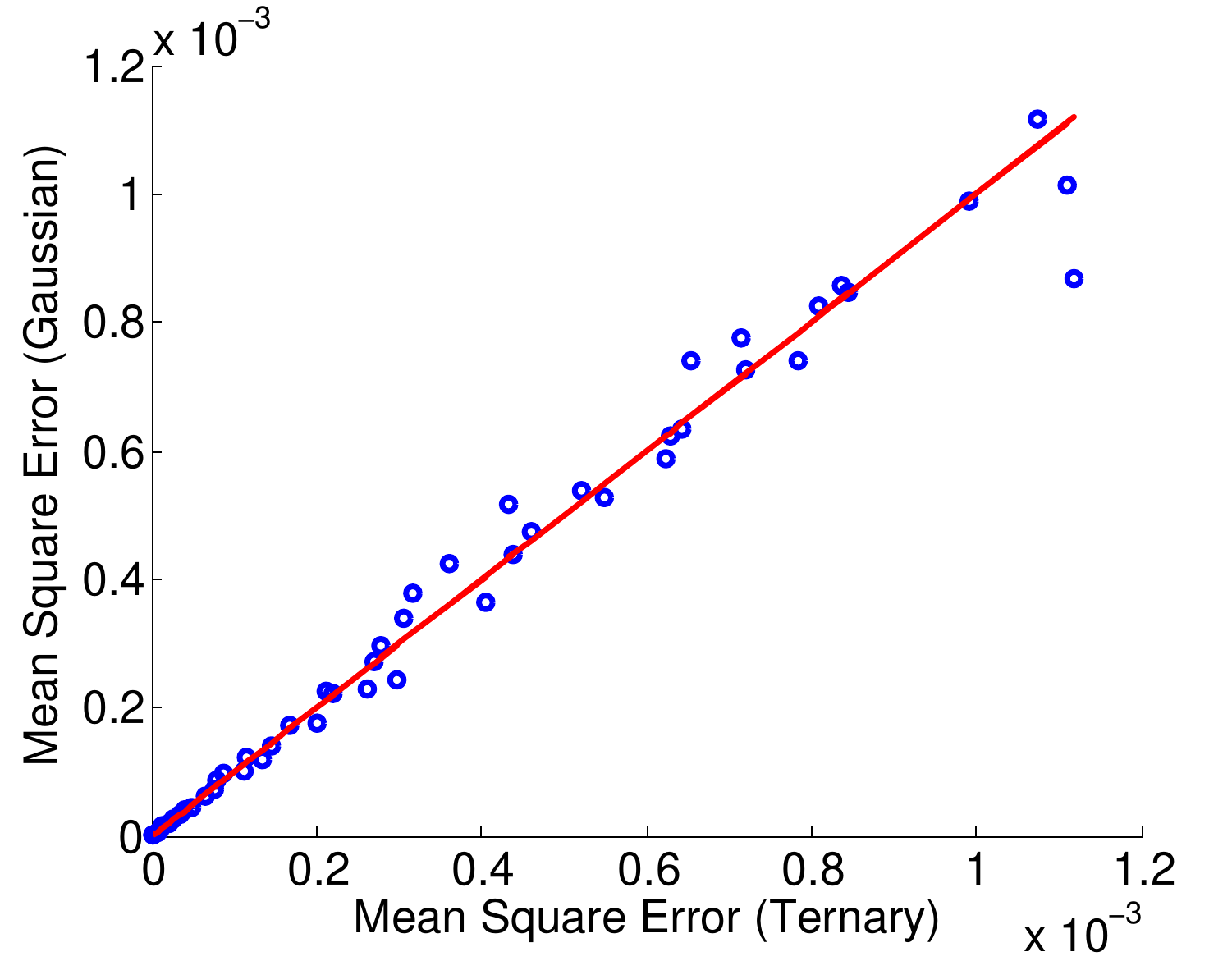}}
\end{center}
\caption{Comparison of the MSE of c-LASSO for Gaussian and
Rademacher matrix ensembles (top), and Gaussian and Ternary ensemble (bottom). The concentration of the points around the $y=x$ line confirms the universality hypothesis. The norms of residuals are $2.8\times10^{-4}$ and $2.3\times10^{-4}$ for the top and bottom figures respectively. Comparison of this figure with Figure \ref{fig:matuniv_noisylasso} confirms that as $N$ grows the data points concentrate more around $y=x$ line.}
\label{fig:matuniv_noisylasso2}
\end{figure}

\begin{figure}
\begin{center}
\subfigure{  \includegraphics[width=5.7cm, height = 5.7cm]{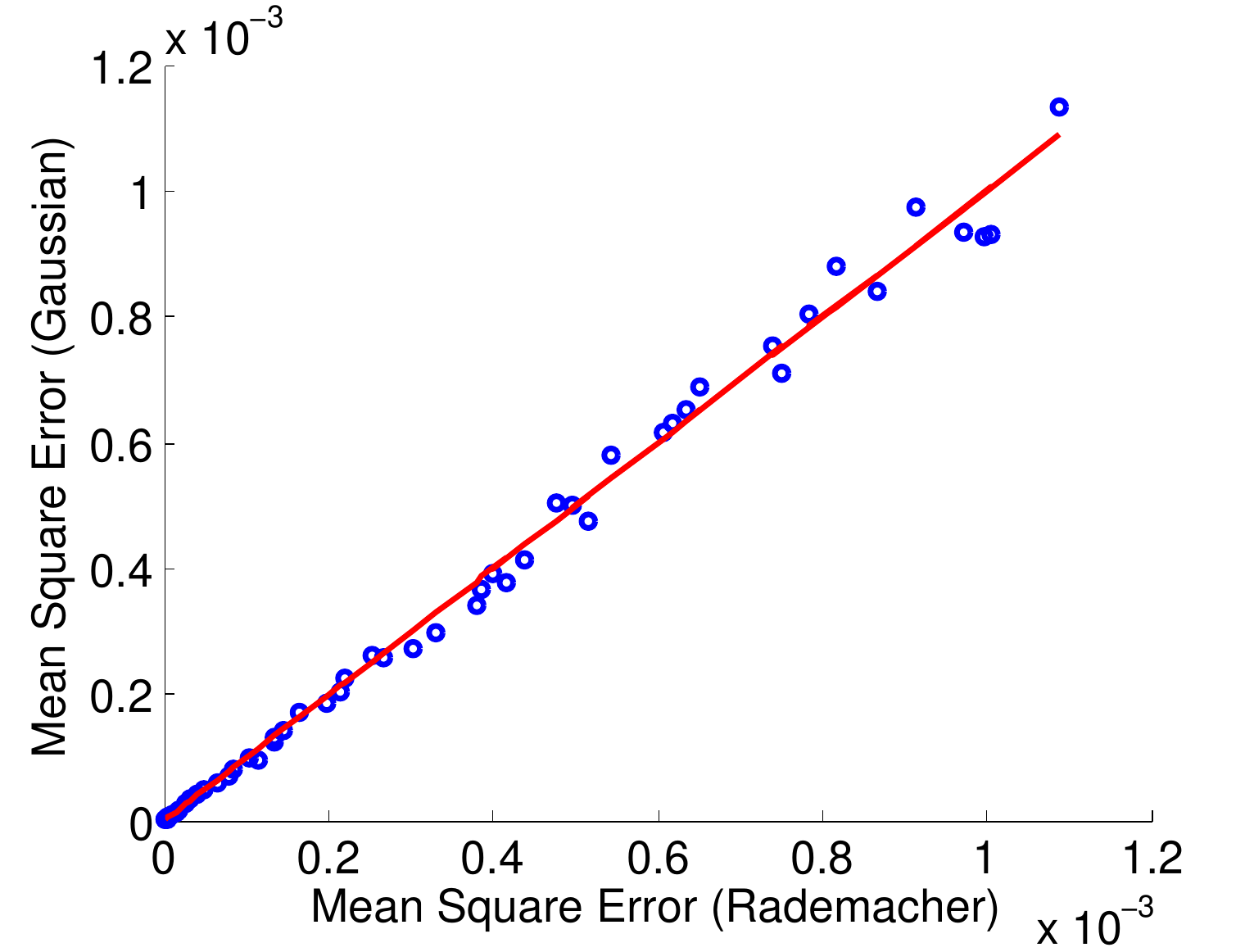}}
\end{center}
\begin{center}
\subfigure{  \includegraphics[width=5.7cm, height= 5.7cm]{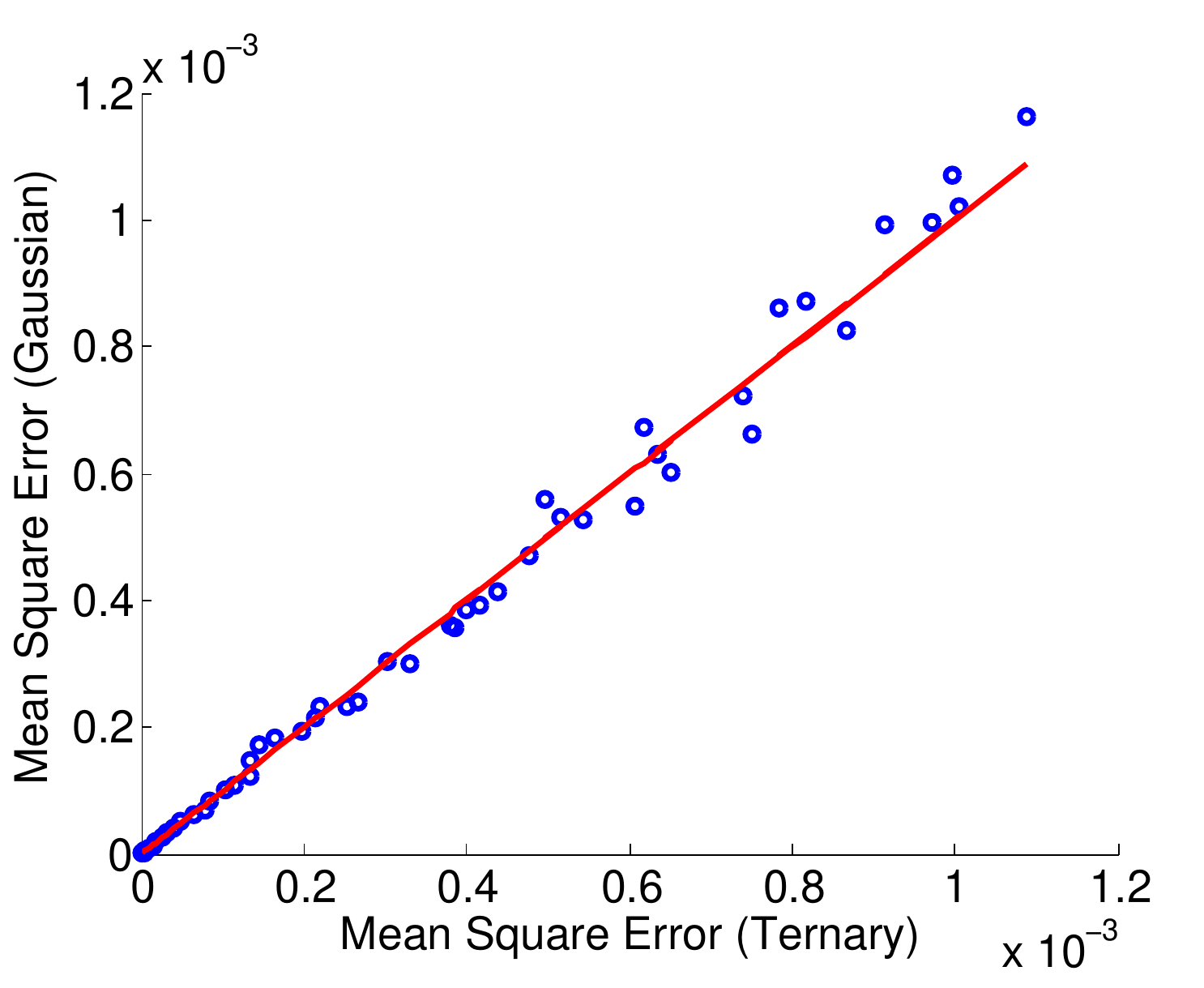}}
\end{center}

\caption{Comparison of the MSE of CAMP for Gaussian and
Rademacher matrix ensembles (top), and Gaussian and Ternary ensemble (bottom). The norms of residuals are $2\times10^{-4}$ and $1.8\times10^{-4}$, respectively. Comparison with Figure \ref{fig:matuniv_noisyCAMP} confirms that as $N$ grows the data points concentrate more around $y=x$ line.}
\label{fig:matuniv_noisyCAMP2}
\end{figure}

\subsection{Coefficient ensemble simulations}\label{sec:coeffuniv}
According to Proposition \ref{prop:distributionindep},
$\rho_{SE}(\delta, \tau)$ is independent of the distribution $G$ of
non-zero coefficients of $s_0$. We test the accuracy of this result
on medium problem sizes. We fix $\delta$ to $0.1$ and we calculate
$\hat{p}^{S}_{\delta, \rho}$ for $60$ equispaced values of $\rho$
between $0.1$ and $0.5$. For each algorithm and each value of $\rho$
we run $100$ Monte Carlo trials and calculate the success rate for
the Gaussian matrix and the coefficient ensembles specified in Table
\ref{table:coeffdist}. Figure \ref{fig:coeffuniv} summarizes our
result. Simulations at other values of $\delta$ result in very
similar behavior. These results are consistent with Proposition
\ref{prop:distributionindep}. The small differences between the empirical phase transitions are due to two issues that
are not reflected in Proposition \ref{prop:distributionindep}: (i) $N$ is finite, while Proposition \ref{prop:distributionindep} considers the asymptotic setting. (ii) The number of algorithm iterations is finite, while  Proposition \ref{prop:distributionindep} assumes that we run CAMP for an infinite number of iterations.

\begin{figure}
\begin{center}
\subfigure{  \includegraphics[width=7cm, height = 6.2cm]{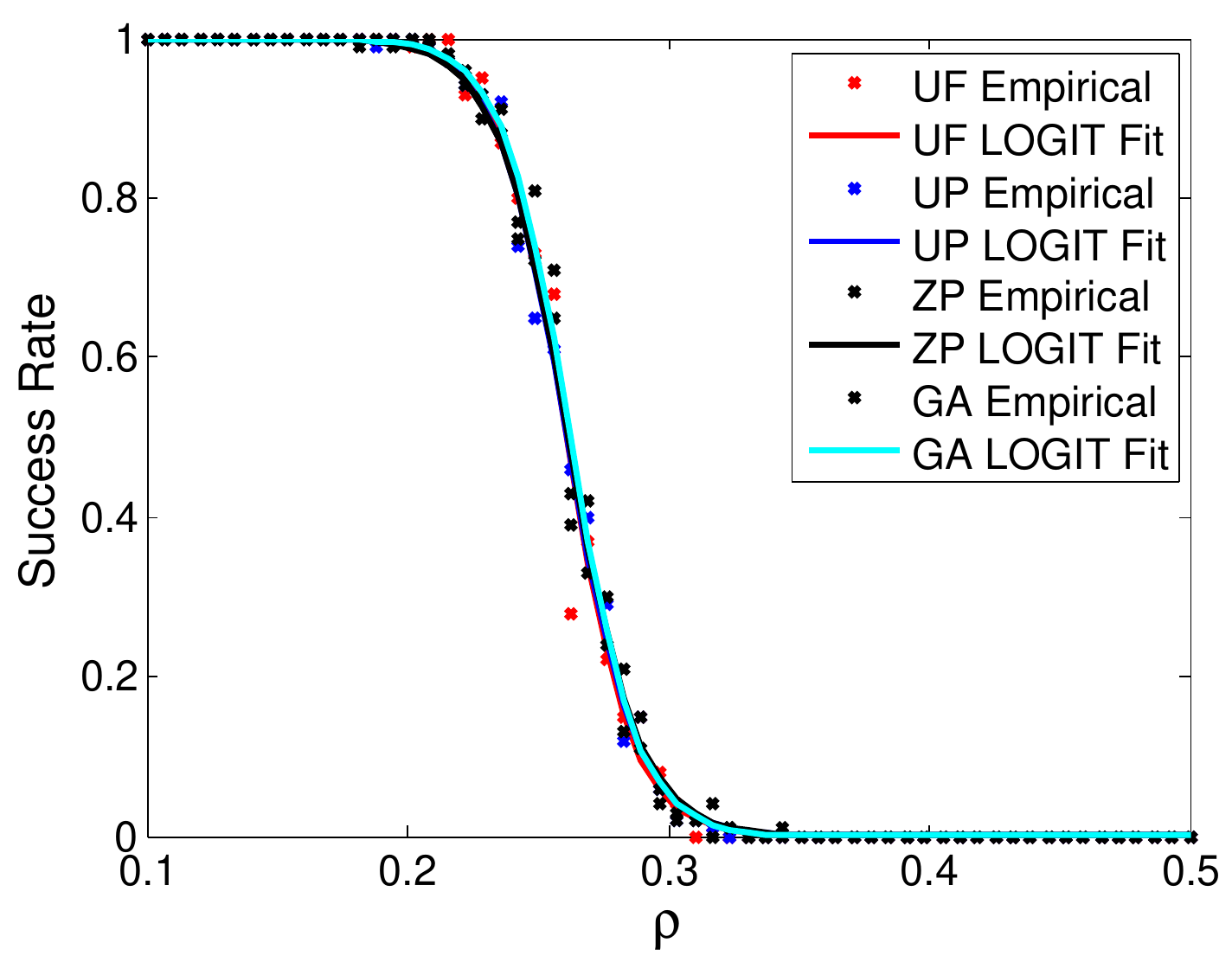}}
\end{center}
\begin{center}
\subfigure{  \includegraphics[width=7cm, height= 6.2cm]{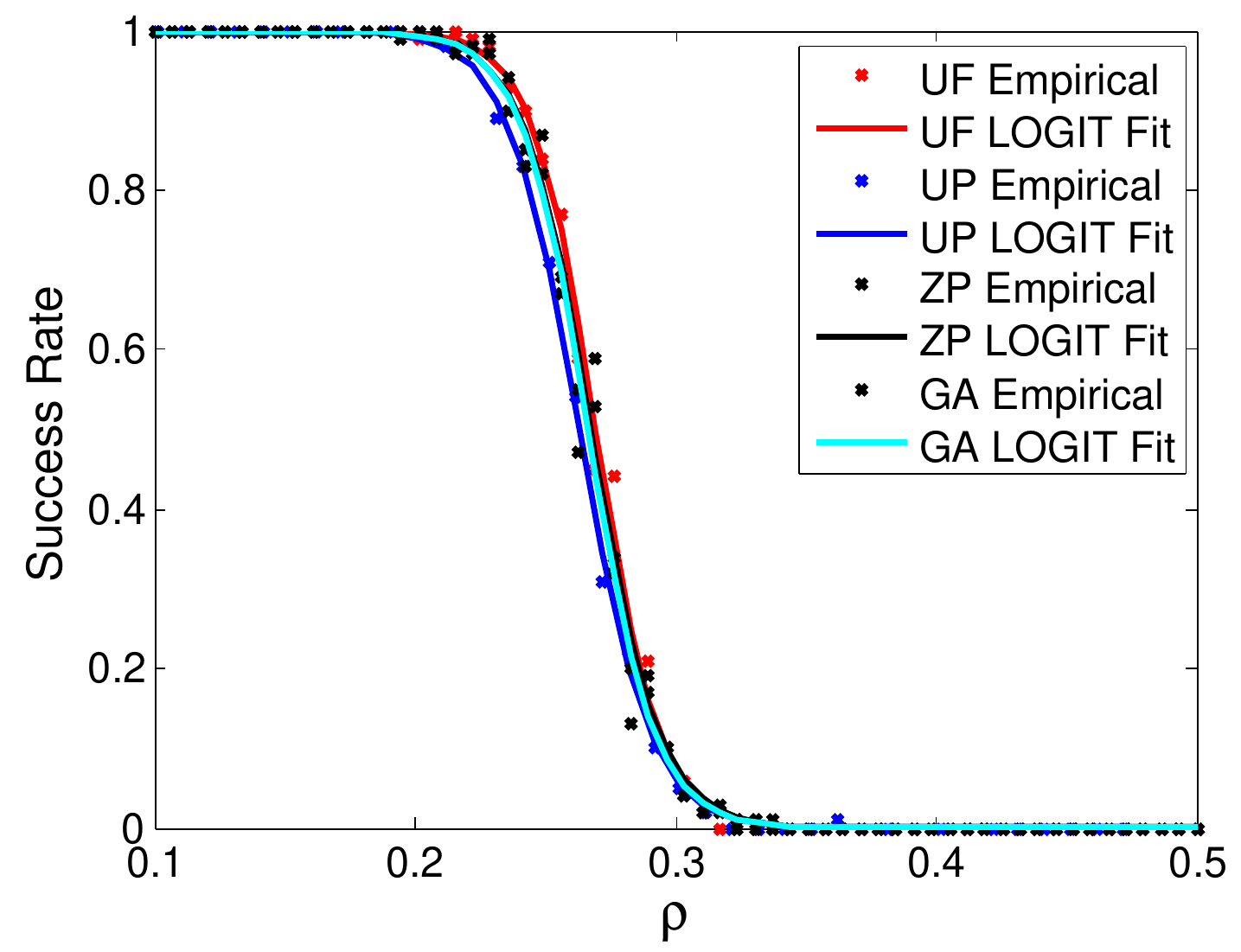}}
\end{center}
\caption{Comparison of the phase transition of c-LASSO (top) and
CAMP (bottom) for different coefficient ensembles specified in Table
\ref{table:coeffdist}. $\delta =0.1$ in this figure. These figures are in agreement with Proposition \ref{prop:distributionindep} that claims the phase transition of CAMP and c-LASSO are independent of the distribution of the non-zero coefficients. Simulations at other values of $\delta$ result in similar behavior.}
\label{fig:coeffuniv}
\end{figure}

\section{Proofs of the main results}
\subsection{Proximity operator}\label{appsec:prox}
For a given convex function $f: \mathds{C}^n\rightarrow \mathds{R}$ the proximity operator at point $x$ is defined as
\begin{equation}\label{eqn:proximity}
{\rm Prox}_{f}(x)\triangleq \arg \min_{y \in\mathds{C}^n} \frac{1}{2} \| y-x\|_2^2 + f(y).
\end{equation}
The proximity operator plays an important role in optimization
theory. For further information refer to \cite{CoWa05} or Chapter 7 of
\cite{MalekiThesis}. The following lemma characterizes the proximity
operator for the complex $\ell_1$-norm. This proximity operator has
been used in several other papers \cite{BeFr08, WrNoFi08, FiNoWr07,
Boyd,YaZh11}.
\begin{lemma}
Let $f$ denote the complex $\ell_1$-norm function, i.e., $f(x) = \sum_i \sqrt{(x_i^R)^2 + (x_i^I)^2}$. Then the proximity
operator is given by
\[
{\rm Prox}_{\tau f} (x) = \eta(x; \tau),
\]
where  $\eta (u+ iv; \tau) = \left(u+iv- \frac{\tau(u+iv)}{\sqrt{u^2+v^2}}\right)_+\mathbb{1}_{\{u^2+v^2> \tau^2\}}$ is applied component-wise
to the vector $x$.
\end{lemma}
\begin{proof}
Since \eqref{eqn:proximity} can be decoupled into the elements of the $x,y$, we can obtain the optimal value of $y$, by optimizing over its individual components. In other words, we solve the optimization in \eqref{eqn:proximity} for $x,y \in \mathds{C}$. In this case the optimization reduces to 
\[
{\rm Prox}_{\tau f} (x) = \arg \min_{y \in \mathds{C}} \frac{1}{2}|y-x|^2+ \tau |y|.
\]
Suppose that the optimal $y_*$ satisfies $(y_*^R)^2+ (y_*^I)^2 >0$. Then the function $\sqrt{(y^R)^2+ (y^I)^2}$ is differentiable and the optimal solution satisfies
\begin{eqnarray} \label{eq:prox2}
x^R-y_*^R  & =& \frac{\tau y_*^R}{\sqrt{(y_*^R)^2+ (y_*^I)^2}} , \nonumber \\
 x^I -y_*^I &= &\frac{\tau y_*^I}{ \sqrt{(y_*^R)^2+ (y_*^I)^2}}.
\end{eqnarray}
Combining the two equations in \ref{eq:prox2} we obtain $y_*^R x^I = x^R y_*^I$. Replacing this in \eqref{eq:prox2} we have $y_*^R = x^R - \frac{\tau |x^R|}{\sqrt{(x^R)^2+ (x^I)^2}}$ and $y_*^I = x^I -\frac{\tau |x^I|}{\sqrt{(x^R)^2+ (x^I)^2}}$.
It is clear that if $\sqrt{(x^R)^2+ (x^I)^2} < \tau$, then the signs of $y_*^R$ and $x^R$ will be opposite, which is in contradiction with \eqref{eq:prox2}. Therefore, if $\sqrt{(x^R)^2+ (x^I)^2} < \tau$, both $y_*^R$ and $y_*^I$ are zero. It is straightforward to check that $(0,0)$ satisfies the subgradient optimality condition.
\end{proof}

\subsection{Proof of Proposition \ref{prop:ampderivation}} \label{appsec:CAMP}
Let
\begin{eqnarray}\label{eq:derivativenotation}
\eta^{R}(x+iy) &\triangleq& \mathcal{R}({\eta(x + iy; \lambda))}, \nonumber \\
\eta^{I}(x+iy) & \triangleq& \mathcal{I}(\eta(x+iy; \lambda))
\end{eqnarray}
denote the real and imaginary parts of the complex soft thresholding function. Define
\begin{eqnarray}\label{eqn:partialderivative}
& & \partial_1\eta^{R}  \triangleq \frac{\partial \eta^{R}(x+iy)}{\partial x}, \nonumber \\
&&\partial_2\eta^{R} \triangleq \frac{\partial \eta^{R}(x+iy)}{\partial y}, \nonumber \\
& &\partial_1 \eta^{I} \triangleq \frac{\partial \eta^{I}(x+iy)}{\partial x},  \nonumber \\
&&\partial_2\eta^{I} \triangleq \frac{\partial \eta^{I}(x+iy)}{\partial y}.
\end{eqnarray}
 We first simplify the expression for $z_{a \rightarrow \ell}^t$:
\begin{eqnarray} \label{eq:zalt}
z_{a\rightarrow \ell}^t &=& \underbrace{y_a - \sum_{j \in [N]} A_{aj} x_j^t - \sum_{j \in [N]} A_{aj} \Delta x_{j \rightarrow a}^t }_{z_a^t \triangleq} \nonumber \\
&+& \underbrace{A_{a\ell} x_\ell^t}_{\Delta z_{a \rightarrow \ell}^t \triangleq} +O(1/N).
\end{eqnarray}
We also use the first-order expansion of the soft thresholding function to obtain
\begin{eqnarray} \label{eq:deltaxia}
\lefteqn{x_{\ell \rightarrow a}^{t+1}} \nonumber \\
&=& \! \! \! \! \eta \Big(\sum_{b \in [n]} A^*_{b\ell} z_{b}^t + \sum_{b \in [n]} A^*_{b\ell} \Delta z_{b \rightarrow \ell}^t -A^*_{a\ell} z_a^t ; \tau_t \Big) \nonumber \\
&& \! \! \! + O\left(\frac{1}{N}\right) \nonumber \\
&=& \! \! \! \!  \underbrace{ \eta \Big(\sum_{b \in [n]} A^*_{b \ell} z_{b}^t + \sum_{b \in [n]} A^*_{b\ell} \Delta z_{b \rightarrow \ell}^t ;\tau_t \Big)}_{x_\ell^t \triangleq}  \nonumber \\
 &&\! \! \! \!  - \, \mathcal{R} (A^*_{a \ell} z_a^t) \partial_1\eta^{R} \Big(\sum_{b \in [n]} A^*_{b\ell} z_{b}^t + \sum_{b \in [n]} A^*_{b\ell} \Delta z_{b \rightarrow \ell}^t \Big) \nonumber \\
 && \! \! \! \! - \, \mathcal{I}(A^*_{a\ell} z_a^t) \partial_2 \eta^{R} \Big(\sum_{b \in [n]} A^*_{b\ell} z_{b}^t + \sum_{b \in [n]} A^*_{b\ell} \Delta z_{b \rightarrow \ell}^t \Big) \nonumber \\
 && \! \! \! \! - \, \mathcal{R}(A^*_{a\ell} z_a^t) \partial_1\eta^{I} \Big(\sum_{b \in [n]} A^*_{b\ell} z_{b}^t + \sum_{b \in [n]} A^*_{b\ell} \Delta z_{b \rightarrow \ell}^t \Big) \nonumber \\
 && \! \! \! \! - \, \mathcal{I} (A^*_{a\ell} z_a^t) \partial_2 \eta^{I} \Big(\sum_{b \in [n]} A^*_{b\ell} z_{b}^t + \sum_{b \in [n]} A^*_{b\ell} \Delta z_{b \rightarrow \ell}^t \Big)  \nonumber \\
 && \! \! \! +O\left(\frac{1}{N}\right). 
\end{eqnarray}
According to \eqref{eq:zalt} $\Delta z_{b \rightarrow \ell}^t = A_{b\ell} x_\ell^t$. Furthermore, we assume that the columns of the matrix are normalized. Therefore, $\sum_b A_{b\ell}^* \delta z_{b \rightarrow \ell}^t = x_\ell^t$. It is also clear that
 \begin{eqnarray}\label{eq:mpproof1}
\lefteqn{ \Delta x_{\ell \rightarrow a}^t} \nonumber \\
 &\triangleq& \! \! \! \!  - \, \mathcal{R}(A^*_{a\ell} z_a^t) \partial_1\eta^{R} \Big(\sum_{b \in [n]} A^*_{b\ell} z_{b}^t + \sum_{b \in [n]} A^*_{b\ell} \Delta z_{b \rightarrow \ell}^t \Big) \nonumber \\
 &&  \! \! \! \!- \, \mathcal{I} (A^*_{a\ell} z_a^t) \partial_2 \eta^{R} \Big(\sum_{b \in [n]} A^*_{bi} z_{b}^t + \sum_{b \in [n]} A^*_{b\ell} \Delta z_{b \rightarrow \ell}^t \Big) \nonumber \\
 && \! \! \! \! - \, \mathcal{R}(A^*_{a\ell} z_a^t) \partial_1\eta^{I} \Big(\sum_{b \in [n]} A^*_{b\ell} z_{b}^t + \sum_{b \in [n]} A^*_{b\ell} \Delta z_{b \rightarrow \ell}^t \Big) \nonumber \\
 && \! \! \! \!- \, \mathcal{I} (A^*_{a\ell} z_a^t) \partial_2 \eta^{I} \Big(\sum_{b \in [n]} A^*_{b\ell} z_{b}^t + \sum_{b \in [n]} A^*_{b\ell} \Delta z_{b \rightarrow \ell}^t \Big). \nonumber\\
 \end{eqnarray}
 Also, according to \eqref{eq:zalt}
 \begin{equation}\label{eq:mpproof2}
 z_a^t = y_a - \sum_j A_{aj} x_j^t - \sum_j A_{aj} \Delta x_{j \rightarrow a}^t.
\end{equation}
By plugging \eqref{eq:mpproof1} into \eqref{eq:mpproof2}, we obtain
\begin{eqnarray*}
\lefteqn{- \sum_{j} A_{aj} \Delta x_{j \rightarrow a}^t} \nonumber \\
&=& \sum_j A_{aj} \mathcal{R}(A^*_{aj} z_a^t) \partial_1\eta^{R}\Big(\sum_b A^*_{bj} z_b^t + x_j^t\Big) \\
&+& \,  \sum_j A_{aj} \mathcal{I}(A^*_{aj} z_a^t) \partial_2 \eta^{R} \Big(\sum_b A^*_{bj} z_b^t + x_j^t \Big) \\
&+&i \,  \sum_j A_{aj} \mathcal{R}(A^*_{aj} z_a^t) \partial_1\eta^{I} \Big(\sum_b A^*_{bj} z_b^t + x_j^t\Big) \\
&+& i \,  \sum_j A_{aj} \mathcal{I}(A^*_{aj} z_a^t) \partial_2 \eta^{I} \Big(\sum_b A^*_{bj} z_b^t + x_j^t\Big) ,
\end{eqnarray*}
which completes the proof. $\hfill \Box$

\subsection{Proof of Lemma \ref{lem:phaseindep}} \label{appsec:proofphaseindep}
Let $\mu$ and $\theta$ denote the amplitude and phase of the random variable $X$. Define $\nu \triangleq \sqrt{\rm npi} = \sqrt{\sigma^2 + \frac{m}{\delta}}$ and $\zeta \triangleq \frac{\mu}{{\nu}}$. Then
\begin{eqnarray}\label{eqn:riskforphase1}
\Psi(m) \! \! \! \! &=& \! \! \! \!  \E|\eta(X + \sqrt{\rm npi} Z_1 + i \sqrt{\rm npi}Z_2; \tau \sqrt{\rm npi})- X|^2  \nonumber \\
             &=&\! \! \! \! \nu^2\E\left|\eta \left(\frac{X}{\nu} + Z_1 + i Z_2; \tau\right)- \frac{X}{\nu} \right|^2  \nonumber \\
           &=& \! \! \! \! (1- \epsilon) \nu^2 \E |\eta( Z_1 + iZ_2; \tau)|^2 \nonumber \\
           &+&\! \! \! \!  {\epsilon} \nu^2 \E \left( \E_{\zeta, \theta} |\eta( \zeta {\rm e}^{i \theta}+Z_1 + iZ_2; \tau)- \zeta {\rm e}^{i\theta}|^2 \right),
\end{eqnarray}
where $\E_{\zeta, \theta}$ denotes the conditional expectation given the variables $\zeta, \theta$.  Note that the marginal distribution of $\zeta$ depends only  on the marginal distribution of $\mu$. The first term in \eqref{eqn:riskforphase1} is independent of the phase $\theta$, and therefore we should
prove that the second term is also independent of $\theta$.  Define
\begin{eqnarray}\label{eqn:riskforphase2}
\Phi(\zeta, \theta) \triangleq \E_{\zeta, \theta}( | \eta(\zeta {\rm e}^{i \theta} + Z_1 + iZ_2  ; \tau) - \zeta {\rm e}^{i \theta}|^2).
\end{eqnarray}
We prove that $\Phi$ is independent of $\theta$. For two real-valued variables $z_r$ and $z_c$, define $\mathbf{z}
\triangleq (z_r, z_c)$, $d \mathbf{z} \triangleq dz_r dz_c$, and
\begin{eqnarray*}
\alpha_z &\triangleq& \sqrt{(\zeta \cos \theta +z_r)^2 + (\zeta \sin \theta +z_c)^2}, \nonumber \\
\chi_z &\triangleq& \arctan\left(\frac{\zeta \sin \theta + z_c}{ \zeta \cos \theta+ z_r} \right), \nonumber\\
c_r &\triangleq& \frac{\zeta \cos \theta +z_r}{\alpha_z}, \nonumber \\
c_i &\triangleq& \frac{\zeta \sin \theta+z_c}{\alpha_z}.
\end{eqnarray*}
Define the two sets $S_{\tau} \triangleq \{(z_r,z_c) \ | \ \alpha_z < \tau \}$  and $S_\tau^c \triangleq \mathds{R}^2 \backslash S_\tau$, where ``$ \backslash$'' is the set subtraction operator. We have
\begin{eqnarray} \label{eq:softthresholdrisk}
\lefteqn{\Phi(\zeta, \theta)} \nonumber \\
&=&\int \limits_{\mathbf{z} \in S_{\tau}} \zeta^2 \frac{1}{\pi} {\rm e}^{-(z_r^2+ z_c^2)} d\mathbf{z} \nonumber \\
&& \! \! \! \! \! \! \! \!+  \int\limits _{\mathbf{z} \in S_{\tau}^c}  \left| (\alpha_z - \tau) {\rm e}^{i \chi_z}- \zeta \cos \theta - i \zeta \sin \theta \right| ^2 \frac{1}{\pi} {\rm e}^{-(z_r^2+ z_c^2)} d\mathbf{z} \nonumber \\
&=&\int \limits_{\mathbf{z} \in S_{\tau}} \zeta^2 \frac{1}{\pi} {\rm e}^{-(z_r^2+ z_c^2)} d\mathbf{z} \nonumber\\
&&\! \! \! \! \! \! \! \!+  \int\limits _{\mathbf{z} \in S_{\tau}^c}  |z_r + iz_c- \tau c_r-i \tau c_i | ^2 \frac{1}{\pi} {\rm e}^{-(z_r^2+ z_c^2)} d\mathbf{z}.
\end{eqnarray}
The first integral in \eqref{eq:softthresholdrisk} corresponds to the case $|\zeta {\rm e}^{i \theta}+z_r + iz_c|< \tau$. The second integral is over the values of $z_r$ and $z_c$ for which $|\zeta {\rm e}^{i \theta}+z_r + iz_c|\geq \tau$. Define $\beta \triangleq \zeta \cos \theta+z_r$ and $\gamma \triangleq \zeta \sin \theta+z_c$. We then obtain
\begin{eqnarray*}
\lefteqn{ \int\limits _{\mathbf{z} \in S_{\tau}^c}  |z_r + iz_c- \tau c_r-i  \tau c_i | ^2 \frac{1}{\pi} {\rm e}^{-(z_r^2+ z_c^2)} dz_r dz_c }\\
& = & \int\limits_{\sqrt{\beta^2+\gamma^2}> \tau} \Big|\beta - \zeta \cos \theta +i(\gamma - \zeta \sin \theta) \nonumber \\
&& \hspace{1.5cm}-\frac{\tau \beta}{\sqrt{\beta^2+\gamma^2}} -i \frac{\tau \gamma}{\sqrt{\beta^2+\gamma^2}} \Big|^2 \nonumber \\
&& \hspace{1.5cm} \frac{1}{\pi} {\rm e}^{-(\beta - \zeta \cos \theta)^2- (\gamma - \zeta \sin \theta)^2} d\beta d\gamma \\
\! \! \! \! &\overset{(a)}{=}&  \int\limits_{\phi =0}^{2 \pi} \int\limits_{r>\tau}  |(r-\tau) \cos \phi - \zeta \cos \theta  \nonumber \\
&&\hspace{1.4cm}+i((r- \tau)\sin \phi - \zeta \sin \theta) | ^2  \nonumber \\
&&\hspace{1.4cm} \frac{1}{\pi} {\rm e}^{-(r \cos \phi - \zeta \cos \theta)^2- (r \sin \phi - \zeta \sin \theta)^2} r dr d\phi \\
&=&   \int\limits_{\phi =0}^{2 \pi} \int\limits_{r>\tau} [(r -
\tau)^2 + \zeta^2 - 2 \zeta (r - \tau) \cos(\theta- \phi)]  \nonumber \\
&&\hspace{1.6cm} {\rm e}^{-
r^2 - \zeta^2 + 2r \zeta \cos(\theta - \phi)}r dr d \phi.
\end{eqnarray*}
Equality (a) is the result of the change of integration variables from $\gamma$ and $\beta$ to $r \triangleq \sqrt{\beta^2+ \gamma^2}$ and $\phi \triangleq \arctan \Big(\frac{\gamma}{\beta} \Big)$. 
The periodicity of the cosine function proves that the last integration is independent of the phase $\theta$. We can similarly prove that $\int_{\mathbf{z} \in S_\tau} \zeta^2 \frac{1}{\pi} {\rm e}^{- z_r^2 +z_c^2} d\mathbf{z}$ is independent of $\theta$. This completes the proof.  $\hfill \Box$

\subsection{Proof of Theorem \ref{thm:phasetranscurve}}\label{appsec:ptformula}

We first prove the following lemma that simplifies the proof of Theorem \ref{thm:phasetranscurve}. 

\begin{lemma}\label{lem:concavity}
The function $\Psi(m)$ is concave with respect to $m$.
\end{lemma}
\begin{proof}
For the notational simplicity define $\nu \triangleq \sqrt{\sigma^2 + \frac{m}{\delta}}$, $X_{\nu} \triangleq \frac{X}{\nu}$, and $A_{\nu} \triangleq |X_{\nu} - Z_1+ i Z_2|$. We note that 
\begin{eqnarray*}
\frac{d^2 \Psi}{dm^2} &=& \frac{d}{dm} \left( \frac{d \Psi}{dm} \right) = \frac{d}{dm} \left( \frac{d\Psi}{d (\nu^2)} \frac{d \nu^2}{dm}\right) \nonumber \\
& =& \frac{1}{\delta}\frac{d}{dm} \left( \frac{d \Psi}{d \nu^2} \right) = \frac{1}{\delta^2} \frac{d^2 \Psi}{ d (\nu^2)^2}.
\end{eqnarray*}
Therefore, $\Psi$ is concave with respect to $m$ if and only if it is concave with respect to $\nu^2$. According to Lemma \ref{lem:phaseindep} the phase distribution of $X$ does not affect the $\Psi$ function. Therefore, we set the phase of $X$ to zero and assume that it is a positive-valued random variable (representing the amplitude). This assumption substantially simplifies the calculations. We have
\begin{eqnarray*}
\Psi(\nu^2) &=& \nu^2 \E \left(\left| \eta(X_{\nu} + Z_1 + iZ_2; \tau) - X_{\nu}\right|^2 \right)  \nonumber \\
                         & = & \nu^2 \E \left(\E_X  \left(\left| \eta(X_{\nu} + Z_1 + iZ_2; \tau) - X_{\nu}\right|^2 \right) \right),
\end{eqnarray*}
 where $\E_X$ denotes the expected value conditioned on the random variable $X$. We first prove that $\Psi_X(\nu^2) \triangleq \nu^2 \E_X  \left(\left| \eta(X_{\nu} + Z_1 + iZ_2; \tau) - X_{\nu}\right|^2 \right)$ is concave with respect to $\nu^2$ by proving $\frac{d\Psi_X}{d( \nu^2)^2} \leq0$. Then, since $\Psi(\nu^2)$ is a convex combination of $\Psi_X(\nu^2)$, we conclude that $\Psi(\nu^2)$ is a concave function of $\nu^2$ as well. The rest of the proof details the algebra required for calculating and simplifying $\frac{d^2 \Psi_X(\nu^2)}{d^2 \nu^2}$. 
 \vspace{.1cm}
 
 Using the real and imaginary parts of the soft thresholding function and its partial derivatives introduced in \eqref{eq:derivativenotation} and \eqref{eqn:partialderivative} we have
\begin{eqnarray*}
\lefteqn{\frac{d\Psi_X(\nu^2)}{d\nu^2}} \nonumber \\
&=& \E_X \left|\eta\left(X_{\nu}+ Z_1+ i Z_2; \tau\right) - X_{\nu}\right|^2 \nonumber \\
&& \! \! \! \!  \! \! \!+ \ \nu^2\frac{d}{d \nu} \E_X \left| \eta\left(X_{\nu}+ Z_1 +iZ_2; \tau \right)- X_{\nu} \right|^2 \frac{d \nu}{d (\nu^2)}\\
&=&  \E_X \left|\eta\left(X_{\nu}+ Z_1+ i Z_2; \tau\right) - X_{\nu}\right|^2 \nonumber \\
&& \! \! \! \! \! \! \! \!+ \frac{\nu}{2}\frac{d}{d \nu} \E_X \left| \eta\left(X_{\nu}+ Z_1 +iZ_2; \tau \right)- X_{\nu} \right|^2 \\
&=&  \E_X \left|\eta\left(X_{\nu}+ Z_1+ i Z_2; \tau\right) - X_{\nu}\right|^2 \nonumber \\
&& \! \! \! \! \! \! \! \!- X_{\nu} \E_X \Big[ \left(\partial_1 \eta^R(X_{\nu} + Z_1 +iZ_2; \tau)-1 \right) \nonumber \\
 &&  \hspace{1cm}\left(\eta^R(X_{\nu} +Z_1+ iZ_2; \tau)- X_{\nu} \right)\Big]\nonumber \\
&& \! \! \! \! \! \! \! \!- X_{\nu} \E_X \Big[ \left(\partial_1 \eta^I(X_{\nu} + Z_1 +iZ_2; \tau) \right)\nonumber \\
 && \hspace{1cm} \left(\eta^I(X_{\nu} +Z_1+ iZ_2; \tau) \right) \Big],
\end{eqnarray*}
where $\partial_1^R$, $\partial_2^R$, $\partial_1^I$, and $\partial_2^I$ are defined in \eqref{eqn:partialderivative}. Note that in the above calculations, $\partial_2 \eta^R$ and $\partial_2 \eta^I$ did not appear, since we assumed that $X$ is a real-valued random variable.  Define $A_\nu \triangleq \sqrt{(X_\nu+ Z_1)^2+ Z_2^2}$. It is straightforward to show that
\begin{eqnarray} \label{eq:defsoft_der}
\partial_1 \eta^R(X_{\nu} + Z_1+ iZ_2 ; \tau) \! \!  \! \! \!&=& \!\! \! \! \! \left(1- \frac{\tau Z_2^2 }{A_\nu^3} \right)\mathds{I}(A_\nu> \tau), \nonumber \\
 \eta^R(X_{\nu} + Z_1 +iZ_2; \tau)\! \! \!  \! \! &=& \! \! \!  \! \! (X_{\nu}+Z_1)\left(\!  1 - \frac{\tau}{A_\nu} \! \right) \mathds{I}(A_\nu \geq \tau), \nonumber \\
 \partial_1 \eta^I(X_{\nu}+Z_1 + iZ_2; \tau)\! \! \! \! \! &=&\! \! \! \! \!  \frac{\tau (X_{\nu} + Z_1)Z_2}{A_\nu^3} \mathds{I}(A_\nu \geq \tau), \nonumber \\
  \eta^I (X_{\nu}+ Z_1 + iZ_2; \tau)\! \! \! \! \!   & = &\! \! \!  \! \!  \left( Z_2 - \frac{\tau Z_2}{A_\nu} \right) \mathds{I} (A_\nu \geq \tau).
\end{eqnarray}

For $f:\mathds{C}\rightarrow \mathds{R}$ we define $\partial^2_1 f(x+iy) \triangleq \frac{\partial^2 f(x+iy)}{\partial x^2}$. It is straightforward to show that
\begin{eqnarray}
\lefteqn{\frac{d^2 \Psi_X(\nu^2)}{d^2 \nu^2} } \nonumber \\
&=& -\frac{X}{\nu^3}\E_X \Big[\left(\partial_1 \eta^R(X_{\nu} + Z_1+ iZ_2 ; \tau)-1 \right) \nonumber \\
&& \hspace{1.5cm} \left( \eta^R(X_{\nu} + Z_1 +iZ_2; \tau)- X_{\nu} \right)\Big] \nonumber\\
&-&  \, \frac{X}{\nu^3} \E_X \Big[ \left(\partial_1 \eta^I(X_{\nu}+Z_1 + iZ_2; \tau) \right) \nonumber \\
&&  \hspace{1.3cm} \left( \eta^I (X_{\nu}+ Z_1 + iZ_2; \tau) \right) \Big]\nonumber \\
&+& \frac{X}{2\nu^3}\E_X \Big[\left(\partial_1 \eta^R(X_{\nu} + Z_1+ iZ_2 ; \tau)-1 \right) \nonumber \\
 &&  \hspace{1.4cm} \left( \eta^R(X_{\nu} + Z_1 +iZ_2; \tau)- X_{\nu} \right) \Big]\nonumber
\end{eqnarray}
\begin{eqnarray}\label{eq:concavity1}
&+&  \, \frac{X}{2\nu^3} \E_X \Big[ \left(\partial_1 \eta^I(X_{\nu}+Z_1 + iZ_2; \tau) \right) \nonumber \\ 
&&  \hspace{1.5cm}  \left( \eta^I (X_{\nu}+ Z_1 + iZ_2; \tau) \right) \Big] \nonumber \\
&+&  \, \frac{X^2}{2\nu^4} \E_X \left( \partial_1 \eta^R(X_{\nu}+ Z_1+ iZ_2;\tau) -1 \right)^2 \nonumber \\
&+& \, \frac{X^2}{2\nu^4} \E_X \left(\partial_1 \eta^I(X_{\nu}+ Z_1+ iZ_2 ; \tau)\right)^2 \nonumber \\
&+&\frac{X^2}{2\nu^4}\E_X \Big[\left(\partial_1^2 \eta^R(X_{\nu} + Z_1+ iZ_2 ; \tau) \right) \nonumber \\ 
&& \hspace{1.4cm} \left( \eta^R(X_{\nu} + Z_1 +iZ_2; \tau)- X_{\nu} \right) \Big]\nonumber\\
&+& \, \frac{X^2}{2\nu^4} \E_X \Big[ \left(\partial_1^2 \eta^I(X_{\nu}+Z_1 + iZ_2; \tau) \right) \nonumber \\
&&  \hspace{1.4cm} \left( \eta^I (X_{\nu}+ Z_1 + iZ_2; \tau) \right) \Big].
\end{eqnarray}
Our next objective is to simplify the terms in \eqref{eq:concavity1}. We start with
\begin{eqnarray}\label{eq:concav2}
&&\E_X \Big[ \left(\partial_1 \eta^R(X_{\nu} + Z_1+i Z_2 ; \tau)-1 \right) \nonumber \\
 && \ \ \ \ \ \ \left( \eta^R(X_{\nu} + Z_1 +iZ_2; \tau)- X_{\nu} \right) \Big]\nonumber \\
&&+ \E_X  \Big[ \left(\partial_1 \eta^I(X_{\nu}+Z_1 + iZ_2; \tau) \right) \nonumber \\
 && \ \ \ \ \ \ \ \ \left( \eta^I (X_{\nu}+ Z_1 + iZ_2; \tau) \right)  \Big]\nonumber \\
&&=   \frac{X}{\nu} \E_X \left( \mathds{I}(A_\nu \leq \tau) +\frac{\tau Z_2^2}{ A_\nu^3} \mathds{I}(A_\nu \geq \tau) \right).
\end{eqnarray}
Similarly,
\begin{eqnarray}\label{eq:concav3}
\lefteqn{ \E_X \left(\partial_1 \eta^R(X_\nu+ Z_1+ iZ_2 ; \tau)-1 \right)^2} \nonumber \\
 &+& \E_X \left(\partial_1 \eta^I(X_\nu + Z_1 + iZ_2; \tau) \right)^2 \nonumber \\
 & = &\E_X \left(\Big(1- \frac{\tau Z_2^2}{A_{\nu}^3}\Big) \mathds{I}(A_{\nu} \geq \tau) -1 \right)^2 \nonumber \\
 &+& \E \left(\frac{\tau(X_\nu +Z_1) Z_2}{A_{\nu}^3} \mathds{I}(A_{\nu} \geq \tau) \right)^2 \nonumber \\
 & = &\E_X \left( \mathds{I}(A_{\nu} \leq \tau) + \frac{\tau^2 Z_2^4}{A_{\nu}^6} \mathds{I} (A_{\nu} \geq \tau) \right) \nonumber \\
 &+& \E_X \left( \frac{\tau^2 (X_\nu + Z_1)^2 Z_2^2}{A_{\nu}^6} \mathds{I}(A_{\nu} \geq \tau) \right) \nonumber \\
 &= & \E_X \left( \mathds{I}(A_{\nu} \leq \tau) + \frac{\tau^2 Z_2^2}{A_{\nu}^4} \mathds{I} (A_{\nu} \geq \tau)  \right). 
 \end{eqnarray}
 We also have
 \begin{eqnarray*}
 \lefteqn{\partial^2_1 \eta^R(X_\nu+ Z_1+ iZ_2 ; \tau)} \nonumber \\
  &=& \frac{3 \tau Z_2^2 (X_\nu+ Z_1)}{A_{\nu}^5} \mathds{I}(A_\nu \geq \tau) \nonumber \\
 &+& \left(1- \frac{\tau Z_2^2}{A_{\nu}^3}\right) \left(\frac{X_\nu + Z_1}{A_{\nu}}\right)\delta(A_\nu - \tau) \nonumber
 \end{eqnarray*}
 and
 \begin{eqnarray*}
 \lefteqn{\partial^2_1 \eta^I(X_\nu+ Z_1+ iZ_2 ; \tau)} \nonumber \\
  &=& \frac{\tau Z_2}{A_\nu^3} \mathds{I}(A_{\nu}\geq \tau) \nonumber \\
 &-& 3\tau\frac{(X_\nu+Z_1)^2Z_2}{A_{\nu}^5} \mathds{I}(A_{\nu}\geq \tau) \nonumber \\
 &+& \frac{\tau (X_\nu+Z_1)^2Z_2}{A_{\nu}^4} \delta(A_\nu- \tau).
 \end{eqnarray*}
 Define
 \begin{eqnarray*}
  S &\triangleq& \E_X \Big[ \left(\partial_1^2 \eta^R(X_{\nu} + Z_1+ iZ_2 ; \tau) \right) \nonumber \\
  &&\ \ \ \ \ \ \left( \eta^R(X_{\nu} + Z_1 +iZ_2; \tau)- X_{\nu} \right) \Big] \nonumber\\ 
 &+& \,  \E_X \Big[ \partial_1^2 \eta^I(X_{\nu}+Z_1 + iZ_2; \tau) \nonumber \\
 &&\ \ \ \ \ \ \eta^I (X_{\nu}+ Z_1 + iZ_2; \tau) \Big]. \nonumber 
 \end{eqnarray*}
 We then have
 \begin{eqnarray}
S&=& \E_X  \left( \frac{3 \tau Z_1 Z_2^2(X_\nu+Z_1)}{A_\nu^5} \mathds{I}(A_\nu\geq \tau) \right)  \nonumber \\ 
&-&\E_X \left( \frac{3 \tau^2 (X_\nu+Z_1)^2Z_2^2}{A_\nu^6} \mathds{I}(A_\nu \geq \tau) \right) \nonumber \\
&-& \E_X\left(  \frac{X_{\nu} (X_{\nu} + Z_1)}{A_{\nu}} \left(1 - \frac{Z_2^2}{A_\nu^2}\right) \delta(A_\nu - \tau) \right) \nonumber \\
&+& \E_X \Big( \Big( \frac{\tau Z_2^2}{A_\nu^3}  - \frac{\tau^2 Z_2^2}{A_{\nu}^4} \Big)\mathds{I}(A_\nu \geq \tau) \Big) \nonumber \\
&-&  \E_X \Big(\frac{3\tau(X_\nu+Z_1)^2 Z_2^2}{A_\nu^5} \Big)\mathds{I}(A_\nu \geq \tau) \Big) \nonumber \\
&-&\E_X \Big( \frac{3 \tau^2(X_\nu+Z_1)^2Z_2^2}{A_\nu^6} \Big)\mathds{I}(A_\nu \geq \tau) \Big). \nonumber
\end{eqnarray}
Note that in the above expression we have replaced $\left(1- \frac{\tau Z_2^2}{A_{\nu}^3}\right)  \delta(A_\nu - \tau)$ with $\left(1 - \frac{Z_2^2}{A_\nu^2}\right) \delta(A_\nu - \tau)$ for an obvious reason.  It is straightforward to simplify this expression to obtain
\begin{eqnarray}\label{eq:secondderivcalc}
S &=& \E_X \left( \left( \frac{\tau Z_2^2}{A_\nu^3} - \frac{\tau^2Z_2^2}{A_\nu^4}  \right)\mathds{I}(A_\nu \geq \tau) \right) \nonumber \\
&-& \E_X \left( \left(\frac{3 \tau (X_\nu+Z_1)Z_2^2 X_\nu}{A_v^5}  \right)  \mathds{I}(A_\nu \geq \tau)  \right) \nonumber \\
&-&  \E_X \left( \frac{X_\nu(X_\nu+ Z_1)}{A_\nu} \left(1- \frac{Z_2^2}{A_\nu^2}  \right) \delta(A_\nu- \tau )\right).
 \end{eqnarray}
 
By plugging \eqref{eq:concav2}, \eqref{eq:concav3}, and \eqref{eq:secondderivcalc} into \eqref{eq:concavity1}, we obtain
\begin{eqnarray}\label{eq:secderivfinal}
\frac{d^2 \Psi_X(\nu^2)}{ d^2 \nu^2} &=& -\E \frac{3 \tau X^3(X_\nu+Z_1)Z_2^2}{2\nu^5 A_\nu^5} \mathds{I}(A_\nu \geq \tau) \nonumber \\
&-&  \E \left(  \frac{X_\nu(X_\nu+ Z_1)}{A_\nu} \right) \left(1 - \frac{Z_2^2}{ A_\nu^2} \right)\delta (A_\nu - \tau). \nonumber \\
\end{eqnarray}
We claim that both terms on the right hand side of \eqref{eq:secderivfinal} are negative. To prove this claim, we first focus on the first term:
\[
\E\left( \frac{(X_\nu+Z_1)Z_2^2}{A_\nu^5} \mathds{I}(A_\nu\geq \tau) \right) \geq 0. 
\]
Define $S_\tau \triangleq \{(Z_1,Z_2) \ | \ A_\nu \geq \tau  \}$. We have
\begin{eqnarray}
\lefteqn{\E\left( \frac{(X_\nu+Z_1)Z_2^2}{A_\nu^5} \mathds{I}(A_\nu\geq \tau) \right)} \nonumber \\
&=& \int \int_{(z_1,z_2) \in S_{\tau}} \frac{(X_\nu+z_1)z_2^2}{A_\nu^5} \frac{1}{\pi} {\rm e}^{-z_1^2- z_2^2} dz_1 dz_2 \nonumber \\
&\overset{(a)}{=}& \int_{\tau}^{\infty} \int_{0}^{2\pi} \frac{r \cos \phi r^2 \sin^2 \phi {\rm e}^{- r^2 - X_\nu^2+2 r X_\nu \cos \phi}}{r^5} r d\phi dr \nonumber \\
&=& \int_\tau^{\infty} \int_0^{2 \pi} \frac{\sin^2 \phi {\rm e}^{- r^2 - X_\nu^2+2 r X_\nu \cos \phi}}{r} d\sin(\phi) dr \geq 0. \nonumber \\
\end{eqnarray}
Equality (a) is the result of the change of integration variables from $z_1$, $z_2$ to $r \triangleq A_\nu$ and $\phi \triangleq \arctan\left( \frac{z_2}{z_1+ X_\nu} \right)$. With exactly similar approach we can prove that the second term of \eqref{eq:secderivfinal} is also negative. 
 
So far we have proved that $\Psi_X(m)$ is concave with respect to $m$. But this implies that $\Psi(m)$ is also concave, since it is a convex combination of concave functions. 
\end{proof}
\vspace{.2cm}

\begin{proof}[Proof of Theorem \ref{thm:phasetranscurve}]
As proved in Lemma \ref{lem:concavity}, $\Psi(m)$ is a concave function. Furthermore $\Psi(0) = 0$.  Therefore a given value of $\rho$ is below the phase transition, i.e., $\rho < \rho_{SE}(\delta)$ if and only if  $\left. \frac{d\Psi}{dm} \right|_m <1$. It is straightforward to calculate the derivative at zero and confirm that
\begin{eqnarray}\label{eq:derivative}
\left. \frac{d\Psi}{d m} \right|_{m = 0} = \frac{\rho \delta(1+ \tau^2)}{ \delta} + \frac{1 -\rho \delta}{ \delta} \E |\eta(Z_1 + iZ_2; \tau) |^2.
\end{eqnarray}
Since $Z_1, Z_2 \sim N(0, 1/2)$ and are independent, the phase of $Z_1+ iZ_2$ has a uniform distribution, while its amplitude has Rayleigh distribution. Therefore, we
have
\begin{eqnarray}\label{eq:risk:nosignal}
\E |\eta(Z_1 + iZ_2; \tau) |^2 = 2\int_{\tau}^{\infty} \omega (\omega - \tau)^2 {\rm e}^{- \omega^2}d\omega.
\end{eqnarray}
We plug \eqref{eq:risk:nosignal} into \eqref{eq:derivative} and set the derivative $\left. \frac{d\Psi}{dm} \right|_m=1$ to obtain the value of $\rho$ at which the phase transition occurs. This value is given by
\begin{eqnarray*}
\rho = \frac{\delta- 2\int_{\tau}^{\infty} \omega (\omega - \tau)^2 {\rm e}^{- \omega^2}d\omega}{\delta(1+ \tau^2 -  2\int_{\tau}^{\infty} \omega (\omega - \tau)^2 {\rm e}^{- \omega^2}d\omega)}.
\end{eqnarray*}
Clearly the phase transition depends on $\tau$. Hence according to the framework we introduced in Section \ref{ssec:phasetran}, we search for the value of $\tau$ that maximizes the phase transition $\rho$. Define $\chi_1(\tau) \triangleq \int_{\tau_*}^{\infty} \omega (\tau_*- \omega){\rm e}^{- \omega^2} d \omega$ and $\chi_2(\tau) \triangleq \int_{\tau_*}^{\infty} \omega(\omega- \tau_*)^2 {\rm e} ^{-\omega^2}d\omega$. This optimal $\tau$ satisfies
\begin{eqnarray*}
\lefteqn{4 \chi_1(\tau^*) \left(1+ \tau_*^2 - 2 \chi_2(\tau^*) \right)} \nonumber \\
& =& \left(4 \chi_1(\tau^*) -2 \tau_{*} \right) \left( \delta - 2\chi_2(\tau^*)\right),
\end{eqnarray*}
which in turn results in $\delta  =  \frac{4(1+ \tau_*^2)
\chi_1(\tau^*)- 4 \tau_* \chi_2(\tau^*)}{-2\tau_* +
4\chi_1(\tau^*) }$. Plugging $\delta$ into the formula for $\rho$,
we obtain the formula in Theorem \ref{thm:phasetranscurve}.
\end{proof}

\subsection{Proof of Theorem \ref{thm:ptasumptote}}\label{appsec:laplace}
We first show that the value of $\delta$ in Theorem \ref{thm:phasetranscurve} goes to zero as $\tau \rightarrow \infty$. By changing the variable of integration from $\omega$ to $\gamma = \omega- \tau$, we obtain
\begin{eqnarray}\label{eq:term11}
\left|\int_{\omega  \geq \tau} \omega (\omega- \tau) {\rm e}^{- \omega^2} d \omega\right| \! &=& \! \left| \int_{\gamma \geq 0} (\gamma+ \tau) \gamma {\rm e}^{- (\gamma+ \tau)^2} d\gamma \right| \nonumber \\ 
\! &\leq& \!  \left|{\rm e}^{- \tau^2} \int_{\gamma \geq 0} (\gamma+ \tau) \gamma {\rm e}^{- \gamma^2} d\gamma \right|  \nonumber \\
\! &=& \! {\rm e}^{- \tau^2} \left(\frac{1}{4} + \frac{\tau}{2 \sqrt{\pi}} \right).
\end{eqnarray}
Again by changing integration variables we have
 \begin{eqnarray}\label{eq:term21}
\left|\int_{\omega \geq \tau} \omega (\omega- \tau)^2 {\rm e}^{- \omega^2} \right| \! &=& \!  \left| \int_{\gamma \geq 0 } (\gamma+ \tau) \gamma^2 {\rm e}^{- (\gamma+\tau)^2 } \right| \nonumber \\
\!&\leq& \! {\rm e}^{- \tau^2} \int_{\gamma> 0} (\gamma+ \tau) \gamma^2 {\rm e}^{- \gamma^2} \nonumber \\
\!&=&\! {\rm e}^{- \tau^2} \left(\frac{\tau}{4} + \frac{1}{2 \sqrt{\pi}}\right).
\end{eqnarray}
Using \eqref{eq:term11} and \eqref{eq:term21} in the formula for $\delta$ in Theorem \ref{thm:phasetranscurve} establishes that $\delta \rightarrow 0$ as $\tau \rightarrow \infty$. Therefore, in order to find the asymptotic behavior of the phase transition as $\delta \rightarrow 0$, we can calculate the asymptotic behavior of $\delta$ and $\rho$ as $\tau \rightarrow \infty$. This is a standard application of Laplace's method. Using this method we calculate the leading terms of $\rho$ and $\delta$:
\begin{eqnarray}\label{eq:asymp_anal1}
\int_{\tau}^{\infty} \omega (\tau - \omega) {\rm e}^{- \omega^2} d\omega \sim \frac{{\rm e}^{- \lambda^2}}{8 \tau^3}, \tau \rightarrow \infty,
\end{eqnarray}
\begin{eqnarray}\label{eq:asymp_anal2}
\int_{\tau}^{\infty} \omega (\tau - \omega)^2 {\rm e}^{- \omega^2} d\omega \sim \frac{{\rm e}^{- \tau^2}}{4 \tau^2}, \tau \rightarrow \infty.
\end{eqnarray}
Plugging \eqref{eq:asymp_anal1} and \eqref{eq:asymp_anal2} into the formula we have for $\rho$ and $\delta$ in Theorem  \ref{thm:phasetranscurve}, we obtain
\begin{eqnarray}
&&\delta \sim \frac{{\rm e}^{- \tau^2}}{2},   \tau \rightarrow \infty, \nonumber \\
&&\rho \sim \frac{1}{\tau^2}, \tau \rightarrow \infty, \nonumber 
\end{eqnarray}
which completes the proof. $\hfill \Box$

\subsection{Proof of Lemma \ref{lem:softthreshincconcave}} \label{app:pf_concavityrisk}
According to Lemma \ref{lem:phaseindep}, the phase $\theta$ does not affect the risk function, and therefore we set it to zero. We have
\begin{eqnarray*}
r(\mu, \tau) &=& \E\left|\eta(\mu + Z_1+ i Z_2; \tau) -\mu \right|^2 \nonumber \\
&=& \E(\eta^R(\mu + Z_1+ iZ_2;\tau) - \mu)^2 \nonumber \\
&& \! \! \! \! \! \! + \ \E(\eta^I(\mu+Z_1+ iZ_2;\tau ))^2, \nonumber
\end{eqnarray*}
where $\eta^R(\mu + Z_1+ iZ_2;\tau) = \left(\mu+ Z_1 -\frac{\tau(\mu+Z_1)}{A} \right)\mathds{I}(A \geq \tau)$,  $\eta^I(\mu + Z_1+ iZ_2;\tau) = \left(z_2 - \frac{\tau Z_2}{A} \right) \mathds{I}(A \geq \tau)$ and
$A \triangleq \sqrt{(\mu+ Z_1)^2 + Z_2^2}$.
If we calculate the derivative of the risk function with respect to $\mu$, then we have
\begin{eqnarray*}
\lefteqn{\frac{d r(\mu,\tau)}{d\mu}} \nonumber \\
&=& 2\E(\eta^R(\mu + Z_1+ iZ_2;\tau) - \mu)\Big(\frac{d \eta^R}{d\mu}-1\Big) \nonumber \\
&+& 2 \E \eta^I(\mu+Z_1 +iZ_2; \tau) \frac{d\eta^I}{d\mu}.\\
\end{eqnarray*}
It is straightforward to show that
\begin{eqnarray*}
\lefteqn{\frac{d r(\mu,\tau)}{d\mu} =  \E \Big[(\eta^R(\mu + Z_1+ iZ_2;\tau) - \mu)} \nonumber \\
&& \hspace{1.2cm} \Big(\Big(1- \frac{\tau Z_2^2}{A^3}\Big)\mathds{I}(A\geq \tau)-1 \Big) \Big] \nonumber \\
&+& \! \! \! \E \Big[\eta^I(\mu+Z_1 +iZ_2; \tau)\Big(\frac{\tau(\mu+Z_1)Z_2}{A^3}\Big) \mathds{I}(A \geq \tau) \Big]\\
&= &\! \! \! \mu \E[\mathds{I}(A \leq \tau)\ \nonumber \\
& -&\! \! \! \E\left[ \left(Z_1 - \frac{\tau(\mu+Z_1)}{A}\right)\left(\frac{\tau Z_2^2}{A^3}\right) \mathds{I}(A\geq \tau)\right] \nonumber \\
&+&\! \! \! \E\left[\left(Z_2- \frac{\tau Z_2}{A}\right)\left( \frac{\tau (\mu+Z_1)Z_2}{A^3}\right) \mathds{I}(A\geq \tau) \right]\\
&= &\! \! \! \mu \E[\mathds{I}(A \leq \tau)] \nonumber \\
&-&\! \! \! \E \left[ \left( \frac{\tau Z_1 Z_2^2}{A^3} + \frac{\tau^2 \mu Z_2^2}{A^4} + \frac{\tau^2 Z_1 Z_2^2}{A^4} \right)  \mathds{I}(A \geq \tau) \right]\nonumber \\
&+& \! \! \! \E \Big[ \Big( \frac{\tau \mu Z_2^2}{A^3} + \frac{\tau Z_2^2 Z_1}{A^3} - \frac{\tau^2 \mu Z_2^2}{A^4} - \frac{\tau^2 Z_1 Z_2}{A^4}\Big) \mathds{I}(A \geq \tau) \Big] \\
 & =&\! \! \! \mu \E[\mathds{I}(A \leq \tau)]+ \mu \E\left[ \frac{\tau Z_2^2}{A^3} \right] \geq 0.
\end{eqnarray*}
Therefore, the risk of the complex soft thresholding is an
increasing function of $\mu$. Furthermore,
\[
2\frac{d r(\mu,\tau)}{d\mu^2} = \frac{1}{\mu}\frac{d r(\mu,\tau)}{d\mu}  = \E(\mathds{I}(A \leq \tau)) +  \E\left( \frac{\tau Z_2^2}{A^3} \right).
\]
It is clear that the next derivative with respect to $\mu^2$ is negative, and therefore the function is concave. $\hfill \Box$

\subsection{Proof of Proposition \ref{prop:minimaxsoft}}\label{app:proofprop}
As is clear from the statement of the theorem, the main challenge here is to characterize
\[
\sup_{q \in F_{\epsilon}} \E |\eta(X+Z_1+iZ_2; \tau)- X |^2.
\]
Let $X = \mu {\rm e}^{i \theta}$, where $\mu$  and $\theta$ are the phase and amplitude of $X$ respectively. According to Lemma \ref{lem:phaseindep}, the risk function is independent of $\theta$. Furthermore, since $q \in F_{\epsilon}$, we can write it as $q (\mu)= (1- \epsilon) \delta_0(\mu) +(1- \epsilon) G(\mu)$, where $G$ is absolutely continuous with respect to Lebesgue measure. We then have
\begin{eqnarray} \label{eq:softmmrisk1}
\lefteqn{\E |\eta(\mu+ Z_1 + i Z_2; \tau)- X |^2} \nonumber \\
& =& (1- \epsilon ) \E|\eta( Z_1 + i Z_2; \tau) |^2 \nonumber \\
&& \! \! \! + \ \epsilon \E_{\mu \sim G}\E_X |\eta(\mu+ Z_1 + i Z_2; \tau)- \mu |^2 \nonumber\\
&=&  2 (1-\epsilon)  \int_{w= \tau}^{\infty} w(w- \tau)^2 {\rm e}^{- w^2}dw \nonumber \\
&& \! \! \! + \ \epsilon \E_{\mu \sim G}\E_\mu |\eta(\mu+ Z_1 + i Z_2; \tau)- \mu |^2.
\end{eqnarray}
The notation $\E_{\mu \tilde G_m}$ means that we are taking the expectation with respect to $\mu$, whose distribution is $G$. Also $\E_\mu$ represents the conditional expectation given the random variable $\mu$. Define $\delta_m(\mu) \triangleq \delta(\mu-m)$. Using Lemma \ref{lem:softthreshincconcave} and the Jensen inequality we prove that $\{G_m(\mu)\}_{m=1}^{\infty}$, $G_m(\mu)=  \delta_m(\mu)$ is the least favorable sequence of distributions, i.e., for any distribution $G$
\begin{eqnarray*}
\lefteqn{ \E_{\mu \sim G}\E_\mu |\eta(\mu+ Z_1 + i Z_2; \tau)- \mu |^2} \nonumber \\
&\leq& \lim_{m \rightarrow \infty} \E_{\mu \sim G_m}\E_\mu |\eta(\mu+ Z_1 + i Z_2; \tau)- \mu |^2.
\end{eqnarray*}
Toward this end we define $\tilde{G}(\mu)$ as $\delta_{\mu_0}(\mu)$ such that $\mu_0^2 = E_G(\mu^2)$. In other words, $\tilde{G}$ and $G$ have the same second moments.
From the Jensen inequality we have
\begin{eqnarray*}
 \lefteqn{\E_{\mu \sim G}\E_\mu |\eta(\mu+ Z_1 + i Z_2; \tau)- \mu |^2 }  \nonumber \\
 &\leq& \! \! \! \E_{\mu \sim \tilde{G}}\E_\mu |\eta(\mu+ Z_1 + i Z_2; \tau)- \mu |^2.
\end{eqnarray*}
Furthermore, from the monotonicity of the risk function proved in Lemma \ref{lem:softthreshincconcave}, we have
\begin{eqnarray*}
\lefteqn{\E_{\mu \sim \tilde{G}}\E_\mu |\eta(\mu+ Z_1 + i Z_2; \tau)- \mu |^2 } \nonumber \\
&\leq& \! \! \!  \E_{\mu \sim G_m}\E_\mu |\eta(\mu+ Z_1 + i Z_2; \tau)- \mu |^2 \ \  \forall m> \mu_0.
\end{eqnarray*}
Again we can use the monotonicity of the risk function to prove that
\begin{eqnarray}\label{eq:softmmrisk2}
\lefteqn{ \E_{\mu \sim G_m}\E_\mu |\eta(\mu+ Z_1 + i Z_2; \tau)- \mu |^2} \nonumber \\
 &\leq& \! \! \!  \lim_{m \rightarrow \infty}    \E_{\mu \sim G_m}\E_\mu |\eta(\mu+ Z_1 + i Z_2; \tau)- \mu |^2 \nonumber \\
 &=&  1+ \tau^2. \hspace{.2 cm}
\end{eqnarray}
The last equality is the result of the monotone convergence theorem. Combining \eqref{eq:softmmrisk1} and \eqref{eq:softmmrisk2} completes the proof. $\hfill \Box$

\section{Conclusions}
We have considered the problem of recovering a complex-valued sparse signal from an undersampled set of complex-valued
measurements. We have accurately analyzed the asymptotic performance
of c-LASSO and CAMP algorithms. Using the state evolution framework, we have derived simple
expressions for the noise sensitivity and phase transition of these two algorithms.
 The results presented here show that
substantial improvements can be achieved when the real and imaginary
parts are considered jointly by the recovery algorithm. For instance,  Theorem \ref{thm:ptasumptote} shows that in the high undersampling regime the phase
transition of CAMP and c-BP is two times higher than the phase transition of r-LASSO.

\section*{Acknowledgements}
Thanks to David Donoho and Andrea Montanari for their encouragement and their valuable suggestions on an early draft of this paper. We would also like to thank the reviewers and the associate editor for the thoughtful comments that helped us improve the quality of the manuscript, and to thank Ali Mousavi for careful reading of our paper and suggesting improvements. This
work was supported by the grants NSF CCF-0431150, CCF-0926127,
and CCF-1117939; DARPA/ONR N66001-11-C-4092 and N66001-11-1-4090; 
ONR N00014-08-1-1112, N00014-10-1-0989, and N00014-11-1-0714;
AFOSR FA9550-09-1-0432; ARO MURI W911NF-07-1-0185 and W911NF-09-1-0383; 
and the TI Leadership University Program.

\bibliographystyle{unsrt}
\bibliography{suthesis}

\ifCLASSOPTIONcaptionsoff
  \newpage
\fi



%

%

\begin{IEEEbiographynophoto}{Arian Maleki} 
received his Ph.D. in electrical engineering
from Stanford University in 2010. After spending 2011-2012 in the DSP group at
Rice University, he joined Columbia University as an Assistant Professor of Statistics. His research interests
include compressed sensing, statistics, machine learning,
signal processing and optimization. He received his M.Sc. in
statistics from Stanford University, and B.Sc. and M.Sc. both in
electrical engineering from Sharif University of Technology.
\end{IEEEbiographynophoto}

\begin{IEEEbiographynophoto}{Laura Anitori}
(S'09) received the M.Sc. degree (summa cum
laude) in telecommunication engineering from the University of Pisa,
Pisa, Italy, in 2005. From 2005 to 2007, she worked as a research
assistant at the Telecommunication department, University of Twente,
Enschede, The Netherlands on importance sampling and its application
to STAP detectors. Since January 2007, she has been with the Radar
Technology department at TNO, The Hague, The Netherlands. Her
research interests include radar signal processing, FMCW and ground
penetrating radar, and airborne surveillance systems. Currently, she
is also working towards her Ph.D. degree in cooperation with the
department of Geoscience and Remote Sensing at Delft University of
Technology, Delft, The Netherlands on the topic of Compressive
Sensing and its applications to radar detection. In 2012, she
received the second prize at the IEEE Radar Conference for the best
student paper award.
\end{IEEEbiographynophoto}


\begin{IEEEbiographynophoto}{Zai Yang}
Zai Yang was born in Anhui Province, China, in 1986. He received the B.S. from mathematics and M.Sc. from applied mathematics in 2007 and 2009, respectively, from Sun Yat-sen (Zhongshan) University, China. He is currently pursuing the Ph.D. degree in electrical and electronic engineering at Nanyang Technological University, Singapore. His current research interests include compressed sensing and its applications to source localization and magnetic resonance imaging.
\end{IEEEbiographynophoto}

\begin{IEEEbiographynophoto}{Richard G. Baraniuk}
 received the B.Sc. degree in 1987 from
the University of Manitoba (Canada), the M.Sc. degree in 1988 from the
University of Wisconsin-Madison, and the Ph.D. degree in 1992 from the
University of Illinois at Urbana-Champaign, all in Electrical
Engineering. After spending 1992--1993 with the Signal Processing
Laboratory of Ecole Normale Sup\'{e}rieure, in Lyon, France, he
joined Rice University, where he is currently the Victor E.\ Cameron
Professor of Electrical and Computer Engineering.   His research
interests lie in the area of signal processing and machine learning.

Dr.\ Baraniuk received a NATO postdoctoral fellowship from NSERC in
1992, the National Young Investigator award from the National
Science Foundation in 1994, a Young Investigator Award from the
Office of Naval Research in 1995, the Rosenbaum Fellowship from the
Isaac Newton Institute of Cambridge University in 1998, the  C.\
Holmes MacDonald National Outstanding Teaching Award from Eta Kappa
Nu in 1999, the University of Illinois ECE Young Alumni Achievement
Award in 2000, the Tech Museum Laureate Award from the Tech Museum
of Innovation in 2006, the Wavelet Pioneer Award from SPIE in 2008,
the Internet Pioneer Award from the Berkman Center for Internet and
Society at Harvard Law School in 2008, the World Technology Network
Education Award and IEEE Signal Processing Society Magazine Column
Award in 2009, the IEEE-SPS Education Award in 2010, the WISE
Education Award in 2011, and the SPIE Compressive Sampling Pioneer
Award in 2012. In 2007, he was selected as one of Edutopia
Magazine's Daring Dozen educators, and the Rice single-pixel
compressive camera was selected by MIT Technology Review Magazine as
a TR10 Top 10 Emerging Technology. He was elected a Fellow of the
IEEE in 2001 and of AAAS in 2009.
\end{IEEEbiographynophoto}



\end{document}